%% file: iRRR_arxiv.tex
\documentclass[12pt]{article}
\usepackage{amsmath}
\usepackage{amssymb}
\usepackage{graphicx}
\usepackage{color}
\usepackage{enumerate}
\usepackage{natbib}
\usepackage{url} 
\usepackage{hyperref} 
\usepackage{float}
\usepackage{amsthm} 
\usepackage{multirow}
\usepackage{subfig} 
\usepackage{algorithm} 
\usepackage{algpseudocode}

\input{myCommand}

\def\trans{^{\rm T}}

\newcommand{\real}{\mathbb{R}} 
\newcommand{\beq}{\begin{equation}}
\newcommand{\eeq}{\end{equation}}
\newtheorem{thm}{Theorem}[section]

\newtheorem*{cor*}{Corollary} 

\newcommand{\Fn}{\mathbb{F}} 
\newcommand{\bes}{\begin{eqnarray*}}
\newcommand{\ees}{\end{eqnarray*}}
\newcommand{\bi}{\begin{itemize}}
\newcommand{\ei}{\end{itemize}}
\newcommand{\bTheta}{\boldsymbol{\Theta}}
\newcommand{\bSigma}{\boldsymbol{\Sigma}}
\newcommand{\bLambda}{\boldsymbol{\Lambda}}

\newcommand{\blind}{1}
\usepackage{authblk}
\usepackage[margin=1in]{geometry}

\usepackage{tikz}
\usetikzlibrary{shapes}
\usetikzlibrary{positioning,arrows}
\usepackage[utf8x]{inputenc}
\usepackage{scalefnt}
\usetikzlibrary{shapes,shapes.geometric,arrows,fit,calc,positioning,automata}
\usetikzlibrary{shapes,arrows,shadows}

\begin{document}

\def\spacingset#1{\renewcommand{\baselinestretch}%
{#1}\small\normalsize} \spacingset{1}


\if1\blind
{
\title{Integrative Multi-View Reduced-Rank Regression: Bridging Group-Sparse and Low-Rank Models}

\author[1]{Gen Li}
\author[2]{Xiaokang Liu}
\author[2,3]{Kun Chen\thanks{Corresponding author. \href{mailto:kun.chen@uconn.edu}{kun.chen@uconn.edu}.}}
\affil[1]{Department of Biostatistics, Columbia University}
\affil[2]{Department of Statistics, University of Connecticut, Storrs, CT}
\affil[3]{Center for Population Health, University of Connecticut Health Center, Farmington, CT}
\date{}
  \maketitle
} \fi

\if0\blind
{
  \bigskip
  \bigskip
  \bigskip
  \begin{center}
    {\LARGE Integrative Multi-View Reduced-Rank Regression: Bridging Group Sparsity and Low-Rank Models}
\end{center}
  \medskip
} \fi

\bigskip
\begin{abstract}
Multi-view data have been routinely collected in various fields of science and engineering. A general problem is to study the predictive association between multivariate responses and multi-view predictor sets, all of which can be of high dimensionality. It is likely that only a few views are relevant to prediction, and the predictors within each relevant view contribute to the prediction collectively rather than sparsely. We cast this new problem under the familiar multivariate regression framework and propose an {\em integrative reduced-rank regression} (iRRR), where each view has its own low-rank coefficient matrix. As such, latent features are extracted from each view in a supervised fashion. For model estimation, we develop a convex \textit{composite nuclear norm penalization} approach, which admits an efficient algorithm via alternating direction method of multipliers. Extensions to non-Gaussian and incomplete data are discussed. Theoretically, we derive non-asymptotic oracle bounds of iRRR under a restricted eigenvalue condition. Our results recover oracle bounds of several special cases of iRRR including Lasso, group Lasso and nuclear norm penalized regression. Therefore, iRRR \textit{seamlessly bridges group-sparse and low-rank methods} and can achieve substantially faster convergence rate under realistic settings of multi-view learning. Simulation studies and an application in the Longitudinal Studies of Aging further showcase the efficacy of the proposed methods.\\



\noindent%
\vfill
\end{abstract}


\spacingset{1.5} 

\pagenumbering{arabic}

\section{Introduction}\label{sec:intro}

Multi-view data, or measurements of several distinct yet interrelated sets of
characteristics pertaining to the same set of subjects, have become increasingly
common in various fields. In a human lung
study, for example, segmental airway tree measurements from CT-scanned images,
patient behavioral data from questionnaires, gene expressions data, together
with multiple pulmonary function test results from spirometry, were all
collected. Unveiling lung disease
mechanisms then amounts to linking the microscopic lung airway structures, the
genetic information, and the patient behaviors to the global measurements of
lung functions \citep{ChenHoffman2016AOAS}. In an Internet network analysis, the popularity and influence of a
web page are related to its layouts, images, texts, and hyperlinks as well as by
the content of other web pages that link back to it. In Longitudinal Study of
Aging (LSOA) \citep{stanziano2010review}, the interest is to predict current health
conditions of patients using historical information of their living conditions,
household structures, habits, activities, medical conditions, among others. The
availability of such multi-view data has made tackling many fundamental problems
possible through an \textit{integrative statistical learning} paradigm, whose success owes to the utilization of
information from various lenses and angles simultaneously. 

The aforementioned problems can all be cast under a multivariate regression framework, in which both the responses and the predictors can be high dimensional, and in addition, the predictors admit some natural grouping structure. In this paper we investigate this simple yet general framework for achieving integrative learning. To formulate, suppose we observe $\bX_k\in \mathbb{R}^{n\times p_k}$ for $k=1,\ldots,K$, each consisting of $n$ copies of independent observations from a set of predictor/feature variables of dimension $p_k$, and also we observe data on $q$ response variables $\Y\in\mathbb{R}^{n\times q}$. Let $\X=(\X_1,\ldots,\X_K)\in\mathbb{R}^{n\times p}$ be the design matrix collecting all the predictor sets/groups, with $p=\sum_{k=1}^{K}p_k$. Both $p$ and $q$ can be much larger than the sample size $n$. Consider the multivariate linear regression model,
\begin{align}
\bY = \bX\bB_0  + \bE = \sum_{k=1}^{K}\bX_k\bB_{0k} + \bE,\label{eq:grmodel}
\end{align}
where $\bB_0=(\bB_{01}\trans,\ldots,\bB_{0K}\trans)\trans\in \mathbb{R}^{p\times q}$ is the unknown regression coefficient matrix partitioned corresponding to the predictor groups, and $\bE$ contains independent random errors with zero mean. For simplicity, we assume both the responses and the predictors are centered so there is no intercept term. 
The naive least squares estimation fails miserably in high dimensions as it leverages neither the response associations nor the grouping of the predictors.




In recent years, we have witnessed an exciting development in regularized estimation, which aims to recover certain parsimonious low dimensional signal from noisy high dimensional data. In the context of multivariate regression or multi-task learning \citep{Caruana1997}, many exploit the idea of sparse estimation \citep{rothman2010sparse,peng2010regularized,lee2012simultaneous,li2015multivariate}, in which information sharing can be achieved by assuming that all the responses are impacted by the same small subset of predictors. When the predictors themselves exhibit a group structure as in model \eqref{eq:grmodel}, a group penalization approach, for example, the convex group Lasso (grLasso) method \citep{yuan2006}, can be readily applied to promote groupwise predictor selection. Such methods have shown to be effective in integrative analysis of high-throughput genomic studies \citep{Ma2011,Liu2012}; a comprehensive review of these methods is provided by \citet{huang2012}.

For multivariate learning, another class of methods, i.e., the reduced-rank methods \citep{anderson1951,reinsel1998}, has also been attractive, where a low-rank constraint on the parameter matrix directly translates to an interpretable latent factor formulation, and conveniently induces information sharing among the regression tasks. \citet{bunea2011optimal} cast the high-dimensional reduced-rank regression (RRR) as a non-convex penalized regression problem with a rank penalty. 
Its convex counterpart is the nuclear norm penalized regression (NNP) \citep{yuan2007,negahban2011,koltchinskii2011},
\begin{align}
\min_{\bB\in \mathbb{R}^{p\times q}}\frac{1}{2n}\|\Y-\X\bB\|_\Fn^2 + \lambda \|\bB\|_\star,\label{eq:NNP}
\end{align}
where $\|\cdot\|_\Fn$ denotes the Frobenius norm, and the nuclear norm is defined as $\|\bB\|_\star = \sum_{j=1}^{p \wedge q}\sigma(\bB,j)$, with $\sigma(\cdot,j)$ denoting the $j$th largest singular value of the enclosed matrix. Other forms of singular value penalization were considered in, e.g., 
\citet{mukherjee2011reduced}, \citet{chen2012ann} and \citet{ZhouLi2014}. In addition, some recent efforts further improve low-rank methods by incorporating error covariance modeling, such as envelope models \citep{cook2015envelopes}, or by utilizing variable selection \citep{chen2012reduced, bunea2012joint, chen2012sparse,su2016sparse}.





In essence, to best predict the multivariate response, sparse methods search for the most relevant subset or groups of predictors, while reduced-rank methods search for the most relevant subspace of the predictors. However, neither class of existing methods can fulfill the needs in the aforementioned multi-view problems. The predictors within each group/view may be strongly correlated, each individual variable may only have weak predictive power, and it is likely that only a few of the views are useful for prediction. Indeed, in the lung study, it is largely the collective effort of the sets of local airway features that drives the global lung functions \citep{ChenHoffman2016AOAS}. In the LSOA study, the predictor groups have distinct interpretations and thus warrant distinct dependence structures with the health outcomes.


In this paper, we propose an {\em integrative multi-view reduced-rank regression} (iRRR) model, where the integration is in terms of multi-view predictors. To be specific, under model \eqref{eq:grmodel}, we assume each set of predictors has its own low-rank coefficient matrix. Figure \ref{fig:irrr} shows a conceptual diagram of our proposed method. Latent features or relevant subspaces are extracted from each predictor set $\X_k$ under the supervision of the multivariate response $\bY$, and the sets of latent variables/subspaces in turn jointly predict $\bY$. The model setting strikes a balance between flexibility and parsimony, as it nicely bridges two seemingly quite different model classes: reduced-rank and group-sparse models. On the one hand, iRRR generalizes the two-set regressor model studied in \citet{velu1991reduced} by allowing multiple sets of predictors, each of which can correspond to a low-rank coefficient matrix. On the other hand, iRRR subsumes group-sparse model setup by allowing the rank of $\bB_{0k}$ being 0, for any $k=1,\ldots, K$, i.e., the coefficient matrix of a predictor group could be entirely zero.


\pgfdeclarelayer{background}
\pgfdeclarelayer{foreground}
\pgfsetlayers{background,main,foreground}
\tikzstyle{xu}=[draw, fill=blue!20, text width=2.5em,
    text centered, minimum height=10em,drop shadow]
\tikzstyle{xv}=[draw, fill=red!20, text width=6em,
    text centered, minimum height=2.5em,drop shadow]
\tikzstyle{ann} = [above, text width=5em, text centered]
\tikzstyle{by} = [xu, text width=6em, fill=red!20,
    minimum height=10em, rounded corners, drop shadow]
\tikzstyle{bx} = [xu, text width=6em, fill=blue!20,
    minimum height=10em, rounded corners, drop shadow]
\tikzstyle{myarrows}=[line width=0.5mm,draw=black,-triangle 45,postaction={draw, line width=1.5mm, shorten >=4mm, -}]
\usetikzlibrary{arrows, decorations.markings}
\tikzstyle{vecArrow} = [thick, decoration={markings,mark=at position
   0.8 with {\arrow[semithick]{open triangle 60}}},
   double distance=1.4pt, shorten >= 12pt,
   preaction = {decorate},
   postaction = {draw,line width=1.4pt, white,shorten >= 11pt}]
\tikzstyle{innerWhite} = [semithick, white,line width=1.4pt, shorten >= 11pt]
\usetikzlibrary{arrows,positioning}
\tikzset{
    >=stealth',
    punkt/.style={
           rectangle,
           rounded corners,
           draw=black, very thick,
           text width=6.5em,
           minimum height=2em,
           text centered},
    pil/.style={
           ->,
           thick,
           shorten <=2pt,
           shorten >=2pt,}
}
\def\blockdist{2}
\def\blockuv{3em}
\def\edgedist{2}
\begin{figure}[htp]
\centering
\begin{tikzpicture}[remember picture,scale=0.6, every node/.style={scale=0.4},every block/.style={scale=0.4}]
    \node at (-1.8,-3) (by) [by]{$\Y$};
    \path (by.east)+(3,0) node (xu1) [xu] {$\X_1\U_1$};
    \path (xu1.east)+(\blockuv,0) node (xv1) [xv] {$\V_1\trans$};
    \path (xv1.east)+(0.4,-0.2) node (dots1)[ann] {$+$};
    \path (dots1.east)+(0,0) node (xu2) [xu] {$\X_2\U_2$};
    \path (xu2.east)+(\blockuv,0) node (xv2) [xv] {$\V_2\trans$};
    \path (xv2.east)+(1,-0.2) node (dots2)[ann] {$+\cdots +$};
    \path (dots2.east)+(0.3,0) node (xuk) [xu] {$\X_K\U_K$};
    \path (xuk.east)+(\blockuv,0) node (xvk) [xv] {$\V_K\trans$};

    \path (xu1.north)+(0,3.6) node (x1) [bx] {$\X_1$};
    \path (xu2.north)+(0,3.6) node (x2) [bx] {$\X_2$};
    \path (xuk.north)+(0,3.6) node (xk) [bx] {$\X_K$};
    \path (x2.east)+(1.5,0) node (dots4) [ann] {$\cdots\cdots$};

\draw [pil,dashed](x1.south)--(xu1.north);
\draw [pil,dashed](x2.south)--(xu2.north);
\draw [pil,dashed](xk.south)--(xuk.north);

\draw [vecArrow](xu1.west)+(-0.6,0)--(by.east);
\draw [innerWhite](xu1.west)+(-0.6,0)--(by.east);

 \draw (by.north)
   edge[pil, red, bend left=15,dashed] (xu1.north) 
   edge[pil, red, bend left=15,dashed] (xu2.north)
   edge[pil, red, bend left=15,dashed] (xuk.north); 

\node [draw=black, fit= (xu1) (xv1) (xuk) (xvk),inner sep=0.25cm,scale=1/0.4] {};

\end{tikzpicture}
\caption{A diagram of integrative multi-view reduced-rank regression (iRRR). Latent features, i.e., $\X_k\U_k$, are learned from each view/predictor set under the supervision of $\Y$.}\label{fig:irrr}
\vspace{-0.35cm}
\end{figure}
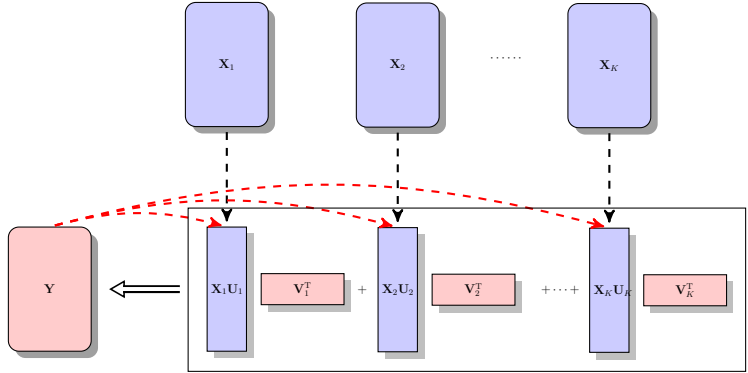

In Section \ref{sec:model}, we develop a new convex optimization approach via composite nuclear norm penalization (cNNP) to conduct model estimation for iRRR, which ensures the scalability to large-scale applications. 
We devise an Alternating Direction Method of Multipliers (ADMM) algorithm to solve the optimization problem with convergence guarantee; extensions to non-Gaussian response, incomplete data, among others, are also considered, and all the details are reported in Section \ref{sec:comp}. In Section \ref{sec:theory}, we derive non-asymptotic oracle bounds for the iRRR estimator, which subsume the results for several existing regularized estimation methods, and show that our proposed approach can achieve superior performance under realistic settings of multi-view learning. Comprehensive simulation studies are contained in Section \ref{sec:sim}, and a real data analysis of the LSOA example is contained in Section \ref{sec:real}. In Section \ref{sec:dis}, we conclude with some discussions.



\section{Integrative Multi-View Reduced-Rank Regression}\label{sec:model}

\subsection{Proposed Model}

We consider the multivariate regression model in \eqref{eq:grmodel} to pursue integrative learning. Recall that in model \eqref{eq:grmodel}, there are $K$ views or groups of predictors denoted by $\bX=(\bX_1,\ldots,\bX_K)$, where $\bX_k\in\real^{n \times p_k}$ and $\sum_{k=1}^K p_k=p$. Correspondingly, the coefficient matrix $\bB_0$ is partitioned into $K$ parts as $\bB_0=(\bB_{01}\trans,\ldots,\bB_{0K}\trans)\trans$, where $\bB_{0k} \in \real^{p_k\times q}$. Denote $r(\cdot)$ as the rank of the enclosed matrix. By assuming each $\bB_{0k}$ is possibly of low rank or even a zero matrix, i.e., $0\leq r_{0k}\ll p_k \wedge q$ where $r_{0k} =r(\bB_{0k})$, for $k=1,\ldots,K$, we reach our proposed {\em integrative multi-view reduced-rank regression} (iRRR) model.





The groupwise low-rank structure in iRRR is distinct from a globally low-rank structure for $\bB_0$ in standard RRR models. The low-rankness of $\bB_{0k}$s does not necessarily imply that $\bB_0$ is of low rank. Conversely, if $\bB_0$ is of low rank, i.e., $r_0=r(\bB_0)\ll p \wedge q$, all we know is that the rank of each $\bB_{0k}$ is upper bounded by $r_0$. 

Nevertheless, we can first attempt an intuitive understanding of the potential parsimony of iRRR in multi-view settings. The numbers of free parameters in $\bB_0$ (the naive degrees of freedom) for an iRRR model, a globally reduced-rank model and a group-sparse model are $\mbox{df}_1 = \sum_{k=1}^K(p_k+q-r_{0k})r_{0k}$, $\mbox{df}_2 = (p+q-r_0)r_0$ and $\mbox{df}_3= \sum_{k=1}^Kp_kqI(r_{0k} \neq 0)$, respectively, where $I(\cdot)$ is an indicator function. For high-dimensional multi-view data, consider the scenario that only a few views/predictor groups impact the prediction in a collective way, i.e., $r_{0k}$s are mostly zero, and each nonzero $r_{0k}$ could be much smaller than $(p_k \wedge q)$. Then $\mbox{df}_1$ could be substantially smaller than both $\mbox{df}_2$ and $\mbox{df}_3$. For example, if $r_{01}>0$ while $r_{0k}=0$ for any $k>1$ (i.e., $r_0=r_{01}$), we have $\mbox{df}_1 = (p_1+q-r_{01})r_{01}$, $\mbox{df}_2 = (p+q-r_{01})r_{01}$ and $\mbox{df}_3=p_1q$, respectively. Another example is when $r_0 = \sum_{k=1}^{K}r_{0k}$, e.g., $\bB_{0k}$s in model \eqref{eq:grmodel} have distinct row spaces. Since $\sum_{k=1}^K (p_k+q-r_{0k})r_{0k}\leq \{q+\sum_{k=1}^K(p_k-r_{0k})\}\{\sum_{k=1}^K r_{0k}\}=(p+q-r_0)r_0$, iRRR is more parsimonious than the globally reduced-rank model. The above observations will be rigorously justified in Section \ref{sec:theory} through a non-asymptotic analysis.

\subsection{Composite Nuclear Norm Penalization}

To recover the desired view-specific low-rank structure in the iRRR model, we propose a convex optimization approach with \textit{composite nuclear norm penalization} (cNNP),
\begin{align}
\widehat{\bB} \in \arg\min_{\bB \in \mathbb{R}^{p \times q}} \frac{1}{2n}\|\bY - \bX\bB\|_{\Fn}^2 + \lambda \sum_{k=1}^{K}w_k \|\bB_k\|_\star, \label{eq:iRRR}
\end{align}
where $\|\bB_k\|_\star = \sum_{j=1}^{p_k \wedge q}\sigma(\bB_k,j)$ is the nuclear norm of $\bB_k$, $w_k$s are some prespecified weights, and $\lambda$ is a tuning parameter controlling the amount of regularization. The use of the weights is to adjust for the dimension and scale differences of $\bX_k$s. We choose
\begin{align}
w_k = \sigma(\bX_k,1)\{\sqrt{q} + \sqrt{r(\bX_k)}\}/n,\label{eq:weight}
\end{align}
based on a concentration inequality of the largest singular value of a Gaussian matrix. This choice balances the penalization of different views and allows us to use only a single tuning parameter to achieve desired statistical performance; see Section \ref{sec:theory} for details.  

Through cNNP, the proposed approach can achieve view selection and view-specific subspace selection simultaneously, which shares the same spirit as the bi-level selection methods for univariate regression \citep{breheny2009,huang2012,ChenHoffman2016AOAS}. Moreover, iRRR seamlessly bridges group-sparse and low-rank methods as its special cases. 

\noindent {\bf Case 1}: \textit{nuclear norm penalized regression (NNP)}. When $p_1=p$ and $K=1$, \eqref{eq:iRRR} reduces to the NNP method as in \eqref{eq:NNP}, which learns a globally low-rank association structure.

\noindent {\bf Case 2}: \textit{multi-task learning (MTL)}. When $p_k=1$ and $p=K$, \eqref{eq:iRRR} becomes a special case of MTL \citep{Caruana1997}, in which all the tasks are with the same set of features and the same set of samples. MTL achieves integrative learning by exploiting potential information sharing across the tasks, i.e., all the task models share the same sparsity pattern of the features.


\noindent {\bf Case 3}: \textit{Lasso and grLasso}. When $q=1$, \eqref{eq:iRRR} becomes a grLasso method, as $\|\bB_k\|_\star = \|\bB_k\|_2$ when $\bB_k\in \mathbb{R}^{p_k}$. Further, when $p_k =1$ and $p=K$, \eqref{eq:iRRR} reduces to a Lasso regression.

\section{Computation and Extensions}\label{sec:comp}

\subsection{ADMM for iRRR}


Without loss of generality, we omit the weights $w_k$ ($k=1,\cdots,K$) defined in \eqref{eq:weight} in the following derivation of the computational algorithm (since we can reparameterize $\bX_k$ by $(1/w_k)\bX_k$ and  $w_k\bB_k$ by $\bB_k$ to get an equivalent unweighted form of the objective function).
The convex optimization has no closed-form solution, for which we propose an ADMM algorithm \citep{boyd2011distributed}.
More specifically, let $\bA_k$ ($k=1,\cdots,K$) be a set of surrogate variables for $\bB_k$ with the same dimensions and $\bA=(\bA_1\trans,\cdots,\bA_K\trans)\trans$.
The original optimization is equivalent to
\bes
\begin{aligned}
\min_{\bA_k,\bB_k}&\quad {1\over 2n}\|\bY - \sum_{k=1}^K \bX_k\bB_k\|^2_\Fn +
  \lambda\sum_{k=1}^K\|\bA_k\|_\star, \qquad s.t.\quad \bA_k=\bB_k, \quad k=1,\cdots,K.
\end{aligned}
\ees
Let $\bLambda_k$ ($k=1,\cdots,K$) be a set of Lagrange multipliers with the same dimensions as $\bA_k$ and $\bB_k$, and $\bLambda=(\bLambda_1\trans,\cdots,\bLambda_K\trans)\trans$.
The augmented Lagrangian objective function is
\begin{eqnarray}\label{obj}
\begin{aligned}
\mathcal{D}(\bY;\bA,\bB, \bLambda)=&{1\over2n}\|\bY-\sum_{k=1}^K \bX_k\bB_k\|^2_\Fn+ \lambda\sum_{k=1}^K\|\bA_k\|_\star \\
&+\ \sum_{k=1}^K\langle\bLambda_k,\bA_k-\bB_k\rangle_\Fn +\ {\rho\over2}\sum_{k=1}^K\|\bA_k-\bB_k\|^2_\Fn,
\end{aligned}
\end{eqnarray}
where $\langle\bQ,\bR\rangle_\Fn$ represents the Frobenius inner product of $\bQ$ and $\bR$, which equals to the trace of $\bQ\trans\bR$.
The last squared Frobenius term is the augmentation term, with $\rho$ being a prespecified step size (usually set to be a small positive value, e.g., 0.1).

The ADMM algorithm alternates between two steps, a {\em primal} step and a {\em dual} step, until convergence. The primal step minimizes $\mathcal{D}(\bY;\bA,\bB,\bLambda)$ with respect to $\bA$ and $\bB$, respectively, while fixing everything else; the {\em dual} step updates $\bLambda$.

\vskip.1in

{\noindent\bf Primal step:}
We minimize \eqref{obj} with respect to $\bA$ and $\bB$, separately.
In particular, when one is fixed, the optimization with respect to the other has an explicit solution.
More specifically, let $\widetilde{\bA}$, $\widetilde{\bB}$, and $\widetilde{\bLambda}$ represent the estimates from the previous iteration.
The optimization $\min_{\bB}\mathcal{D}(\bY;\widetilde{\bA},\bB,\widetilde{\bLambda})$ has a unique solution
\beq\label{B}
\widehat{\bB}=\left({1\over n}\bX\trans\bX+\rho\bI\right)^{-1}\left({1\over n}\bX\trans\bY+\rho\widetilde{\bA}+\widetilde{\bLambda}\right).
\eeq
Subsequently, we can obtain the estimate of $\bB_k$ (i.e.,  $\widehat{\bB}_k$) by partitioning $\widehat{\bB}$.

To estimate $\bA$, the objective function $\mathcal{D}(\bY;\bA,\widehat{\bB}; \widetilde{\bLambda})$ is readily separable for different $\bA_k$s.
In particular, each subproblem is rewritten as
\begin{align}
\min_{\bA_k} \quad {\rho\over 2} \|\bA_k-\widehat{\bB}_k+{\widetilde{\bLambda}_k\over\rho}\|^2_\Fn +\lambda\|\bA_k\|_\star,\label{eq:Astep}
\end{align}
which can be solved via the singular value soft-thresholding technique \citep{cai2010singular}.
To be specific, let $\bU_k\bD_k\bV_k\trans$ be the singular value decomposition (SVD) of $\widehat{\bB}_k-\widetilde{\bLambda}_k/\rho$, where $\bU_k$ and $\bV_k$ have orthonormal columns and $\bD_k$ contains non-increasing singular values.
The solution to the optimization problem in \eqref{eq:Astep} is
\beq\label{A}
\widehat{\bA}_k=\bU_k\mathcal{S}(\bD_k,{\lambda\over\rho})\bV_k\trans,
\eeq
where $\mathcal{S}(\bD_k,\lambda/\rho)=(\bD_k-\lambda/\rho)_+$ applies soft-thresholding at the level $\lambda/\rho$ to each entry of $\bD_k$.
As a result, $\widehat{\bA}_k$ may be low-rank.

\vskip.1in
{\noindent\bf Dual step:}
Once $\bA$ and $\bB$ are estimated, the Lagrange multipliers $\bLambda_k$ are updated by
\beq\label{dual}
\widehat{\bLambda}_k=\widetilde{\bLambda}_k+\rho(\widehat{\bA}_k-\widehat{\bB}_k).
\eeq


\vskip.1in
{\noindent\bf Stopping criterion:}
The ADMM algorithm alternates between the primal step and the dual step. After each iteration, we evaluate the primal and dual residuals as
\beq\label{res}
r_{primal}=\|\widehat{\bA}-\widehat{\bB}\|_\Fn,\quad r_{dual}=\rho\|\widehat{\bB}-\widetilde{\bB}\|_\Fn.
\eeq
Following \cite{boyd2011distributed}, the stopping criterion is that both residuals fall below a small prefixed threshold. It can be proved that under weak regularity conditions, the algorithm always converges to a global optimum. In practice, one can let the step size $\rho$ vary over iterations, and generally the convergence is expedited with a slowly increasing sequence of $\rho$ \citep{he2000alternating}. 
A summary of the above algorithm for solving iRRR with a fixed $\lambda$ is provided in Algorithm \ref{alg1} below. 

\begin{algorithm}[h]
\caption{ADMM algorithm for fitting iRRR}\label{alg1}
\begin{algorithmic}
\State Parameter: $\lambda$, $\rho$.
\State Initialize $\bA$, $\bB$ and the Lagrange multiplier $\bLambda$;
\While {The stopping criterion is not satisfied}
\bi
\item Primal step: update $\bB$ by \eqref{B} and update $\bA$ by \eqref{A};
\item Dual step: update $\bLambda$ by \eqref{dual};
\item Calculate the primal and dual residuals in \eqref{res};
\item (Optional) Increase $\rho$ by a small amount, e.g., $\rho \leftarrow 1.1\rho$.
\ei
\EndWhile
\end{algorithmic}
\end{algorithm}


The tuning parameter $\lambda$ in \eqref{eq:iRRR} balances the loss function and the penalty term. In practice, the model is fitted using the ADMM algorithm for a sequence of $\lambda$ values to produce a spectrum of view-specific low-rank models. A warm start strategy is adopted to speed up computation, i.e., the current solution is used as the initial value for the next $\lambda$ value. We use K-fold cross validation \citep{stone1974} to choose the optimal $\lambda$ and hence the optimal solution, based on the predictive performance of the models.


\subsection{Handling Non-Gaussian and Incomplete Response}


When the responses are non-Gaussian, we substitute the squared loss function in \eqref{eq:iRRR} with the negative log likelihood denoted as $-\log L(\bY,\bTheta)$. The augmented Lagrangian becomes
\bes
\mathcal{D}(\bY;\bmu,\bA,\bB, \bLambda)=-{1\over n}\log L(\bY,\bTheta)+ \lambda\sum_{k=1}^K\|\bA_k\|_\star
+ \langle\bLambda,\bA-\bB\rangle_\Fn + {\rho\over2}\|\bA-\bB\|^2_\Fn,
\ees
where $\bTheta=\1\bmu\trans+\bX\bB$.
The minimization of $\mathcal{D}(\bY;\bmu,\bA,\bB, \bLambda)$ with respect to $\bmu$ and $\bB$ while fixing everything else may no longer have closed-form solutions.
To alleviate the computational burden, one could apply a quadratic approximation or majorization to the negative log likelihood function in the primal step, and then follow the ADMM algorithm for parameter estimation.
In the following, we demonstrate the estimation procedure for binary responses.

The log-likelihood function for binary responses $\bY$ can be expressed as
\beq\label{binary}
\log L(\bY,\bTheta)=\sum_{i=1}^n\sum_{j=1}^q \log h\left((2y_{ij}-1)\theta_{ij}\right),
\eeq
where $\theta_{ij}$ is the $(i,j)$th entry of $\bTheta$ and $h(\eta)=\exp(\eta)/\{1+\exp(\eta)\}$ denotes the inverse function of the logit link function.
Following \cite{lee2010sparse} and \cite{lee2013coordinate}, we have the following relation
\beq\label{major}
-\log h(\eta)\leq -\log h(\eta_0)- 2\{1-h(\eta_0)\}^2+{1\over 8}\left[\eta-\eta_0-4\{1-h(\eta_0)\}\right]^2.
\eeq
Namely, $-\log h(\eta)$ is majorized by the quadratic function on the right-hand side, which is tangent with $-\log h(\eta)$ at $\eta_0$ and has a fixed second-order derivative.

Let $\widetilde{\theta}_{ij}$ be the estimate from the previous iteration.
By applying \eqref{major} to \eqref{binary}, we have
\[
-\log L(\bY,\bTheta)\leq {1\over8}\sum_{i=1}^n\sum_{j=1}^q\left[(2y_{ij}-1)(\theta_{ij}-\widetilde{\theta}_{ij})
-4\left\{1-h\left((2y_{ij}-1)\widetilde{\theta}_{ij}\right)\right\}\right]^2 + c,
\]
where $c$ is some constant.
Let $\bY^\star$  be an $n\times q$ working response matrix with the $(i,j)$th entry
\bes
y_{ij}^\star=\widetilde{\theta}_{ij}+4(2y_{ij}-1)\left\{1-h\left((2y_{ij}-1)\widetilde{\theta}_{ij}\right)\right\}.
\ees
Correspondingly, the negative log likelihood function $-\log L(\bY,\bTheta)$ is majorized by the squared function
${1/8}\|\bY^\star-\bTheta\|^2_\Fn$, plus some constant.
Consequently, in the primal step, one could minimize the majorized objective function to estimate $\bmu$ and $\bB$ explicitly. In particular, the estimate of $\bmu$ is $({1/n}){\bY^\star}\trans\1$.
We remark that in practice, it generally suffices to run the majorization-minimization procedure once in each ADMM iteration \citep{he2002new}.
When there are missing values in the responses, we exploit a similar idea to majorize the objective function in each ADMM iteration.
More specifically, suppose $\mathcal{O}\subseteq\{(i,j): i=1,\cdots,n; j=1,\cdots,q\}$ is the index set for observed data points, and $\mathcal{M}\subseteq\{(i,j): i=1,\cdots,n; j=1,\cdots,q\}$ is the index set for missing values.
For Gaussian data, we  majorize the observed loss function $\sum_{(i,j)\in\mathcal{O}}(y_{ij}-\theta_{ij})^2$ by $\sum_{(i,j)\in\mathcal{O}}(y_{ij}-\theta_{ij})^2+\sum_{(i,j)\in\mathcal{M}}(\widetilde{\theta}_{ij}-\theta_{ij})^2$;
for binary responses, we first majorize the negative log likelihood function by $1/8\sum_{(i,j)\in\mathcal{O}}(y^\star_{ij}-\theta_{ij})^2$ as before, and then further majorize it as
$1/8\sum_{(i,j)\in\mathcal{O}}(y^\star_{ij}-\theta_{ij})^2+1/8\sum_{(i,j)\in\mathcal{M}}(\widetilde{\theta}_{ij}-\theta_{ij})^2$.
By collecting $y_{ij}$ or $y^\star_{ij}$ and $\widetilde{\theta}_{ij}$ in an $n\times p$ matrix, we obtain a matrix-form loss function as before.
As a result, we use the same ADMM steps to estimate the parameters.

\subsection{On $\ell_2$ Regularization and Adaptive Estimation}

To better deal with high dimensional data, we can consider adding a ridge penalty $\lambda_2 \|\bB\|_{\Fn}^2$ to the cNNP penalty in \eqref{eq:iRRR} \citep{mukherjee2011reduced, chen2012ann}.
As a result, the objective function becomes strictly convex whenever the tuning parameter $\lambda_2 > 0$. This shares the same idea as the elastic net \citep{zou2005}, and ensures that the problem has a unique global optimizer.


With the combined penalty form $\lambda \sum_{k=1}^{K}\|\bB_k\|_\star + \lambda_2 \|\bB\|_{\Fn}^2$, the iRRR problem can be easily transformed to the same form as before:
\[
{1\over2n}\left\|\begin{pmatrix}\bY\\ \0\end{pmatrix}-\begin{pmatrix} \bX\\ \sqrt{2n\lambda_2}\bI\end{pmatrix}\bB\right\|^2_\Fn+ \lambda\sum_{k=1}^K\|\bB_k\|_\star,
\]
where $\0$ is a zero matrix of size $p\times q$ and $\bI$ is an identity matrix of size $p\times p$. (More generally the identity matrix can be replaced by a diagonal matrix to allow weighted $\ell_2$ regularization). Upon defining $\bY^\dag=(\bY\trans,\0)\trans$ and $\bX^\dag=(\bX\trans,\sqrt{2n\lambda_2}\bI)\trans$ as augmented responses and predictors, the model estimation could be conducted directly by applying Algorithm \ref{alg1} to the augmented data. Alternatively, a more computationally efficient way is to directly modify the ADMM algorithm by replacing the nuclear norm penalty in \eqref{eq:Astep} by a combined nuclear  and squared $\ell_2$ norm penalty. The resulting problem can still be solved explicitly, now via a singular value shrinkage and thresholding operation \citep{Sun2012}. 

When the ridge penalty is included, we have an additional tuning parameter $\lambda_2$. A larger value of $\lambda_2$ makes the problem more convex, but meanwhile introduces more bias to the final estimates. In practice, $\lambda_2$ can be selected using CV as well. However, empirical experiments suggest that it usually suffices to set $\lambda_2$ at a very small value without tuning it.
For simplicity, we omit the ridge penalty term in our numerical studies.

Moreover, motivated by \citet{zou2006adaptive}, we can consider an adaptively weighted version of iRRR, where, for example, we first fit iRRR and then adjust the weights according to the estimated coefficient sub-matrices (e.g., factoring in the inverse of the Frobenius norms of the estimated coefficient matrices). This may potential improve view selection and predictive accuracy, as shown in the numerical studies in Section \ref{sec:sim}.

\section{Theoretical Analysis}\label{sec:theory}

We investigate the theoretical properties of the proposed iRRR estimator from solving the convex cNNP problem. In particular, we derive its non-asymptotic performance bounds for estimation and prediction. Our general results recover performance bounds of several related methods, including Lasso, grLasso and NNP. We further show that iRRR is capable of substantially outperforming those methods under realistic settings of multi-view learning. All the proofs are provided in Appendix \ref{sec:appendix:th}.

We mainly consider the multi-view regression model in \eqref{eq:grmodel}, i.e.,
$\bY = \sum_{k=1}^{K}\bX_k\bB_{0k} + \bE$, and the iRRR estimator in \eqref{eq:iRRR} with the weights defined in \eqref{eq:weight}, i.e.,
\begin{align*}
\widehat{\bB} \in \arg\min_{\bB\in \mathbb{R}^{p\times q}}\ \frac{1}{2n}\|\bY - \bX\bB\|_\Fn^2 + \lambda \sum_{k=1}^{K}\sigma(\bX_k,1)\left\{\sqrt{q} + \sqrt{r(\bX_k)}\right\} \|\bB_k\|_\star/n.
\end{align*}
Define $\bZ = \bX\trans\bX/n$, and $\bZ_k = \bX_k\trans\bX_k/n$, for $k=1,\ldots,K$. We scale the columns of $\bX$ such that the diagonal elements of $\bZ$ all equal to 1. Denote $\Lambda(\bZ,l)$ as the $l$th largest eigenvalue of $\bZ$, so that $\Lambda(\bZ,l) = \sigma(\bX,l)^2/n$.

\begin{thm}\label{th:1}
Assume $\bE$ has independent and identically distributed (i.i.d.) $\mbox{N}(0,\tau^2)$ entries. Let $\lambda = (1+\theta)\tau$, with $\theta >0$ arbitrary. Then with probability at least
$1 - \sum_{k=1}^{K}\exp[-\theta^2\{q+r(\bX_k)\}/2]$,
we have
\begin{align*}
\|\bX\widehat{\bB} - \bX\bB_0\|_\Fn^2
\leq \|\bX\bC - \bX\bB_0\|_\Fn^2 + 4\lambda \sum_{k=1}^{K}\sigma(\bX_k,1)\left\{\sqrt{q}+\sqrt{r(\bX_k)}\right\}\|\bC_k\|_\star,
\end{align*}
for any $\bC_k\in \mathbb{R}^{p_k\times q}$, $k=1,\ldots,K$ and $\bC=(\bC_1\trans,\ldots, \bC_K\trans)\trans$.
\end{thm}

Theorem \ref{th:1} shows that $\widehat{\bB}$ balances the bias term $\|\bX\bC - \bX\bB_0\|_\Fn^2$ and the variance term $4\lambda \sum_{k=1}^{K}\sigma(\bX_k,1)\{\sqrt{q}+\sqrt{r(\bX_k)}\}\|\bC_k\|_\star$. An oracle inequality for $\widehat{\bB}$ is then readily obtained for the low-dimensional scenario $\sigma(\bX, p) > 0$; see the corollary in Appendix \ref{sec:supp:th1}.



We now investigate the general high-dimensional scenario. Motivated by \citet{lounici2011}, \citet{negahban2011}, \citet{koltchinskii2011}, among others, we impose a restricted eigenvalue condition (RE). We say that $\bX$ satisfies RE condition over a restricted set $\mathcal{C}(r_1,\ldots, r_K;\delta) \subset \mathbb{R}^{p\times q}$ if there exists some constant $\kappa(\bX) > 0 $ such that
\begin{align*}
\frac{1}{2n}\|\bX\Delta\|_\Fn^2 \geq \kappa(\bX)\|\Delta\|_\Fn^2, \qquad \mbox{ for all } \Delta \in \mathcal{C}(r_1,\ldots, r_K;\delta).
\end{align*}
Here each $r_k$ is an integer satisfying $1\leq r_k \leq \min(p_k,q)$ and $\delta$ is a tolerance parameter. The technical details on the construction of the restricted set is provided in Appendix \ref{sec:th:rec}.



\begin{thm}\label{th:2}
Assume that $\bE$ has i.i.d.\ $\mbox{N}(0,\tau^2)$ entries. Suppose $\bX$ satisfies the RE condition with parameter $\kappa(\bX) >0$ over the set $\mathcal{C}(r_1,\ldots, r_K;\delta)$. Let $\lambda=2(1+\theta) \tau$ with $\theta>0$ arbitrary. Then with probability at least
$1 - \sum_{k=1}^{K}\exp[-\theta^2\{q+r(\bX_k)\}/2]$, 
\begin{align*}
\|\widehat{\bB} - \bB_{0}\|_\Fn^2  \preceq  \max & \left\{  \delta^2, \tau^2(1+\theta)^2\sum_{k=1}^{K}\frac{\Lambda(\bZ_k,1)}{\kappa(\bX)^2}\frac{\{\sqrt{q} + \sqrt{r(\bX_k)}\}^2r_k}{n}, \right.\\
& \left.\tau(1+\theta)\sum_{k=1}^{K}\frac{\sqrt{\Lambda(\bZ_k,1)}}{\kappa(\bX)}\frac{\{\sqrt{q} + \sqrt{r(\bX_k)}\}\{\sum_{j=r_k+1}^{m_k}\sigma(\bB_{0k},j)\}}{\sqrt{n}}
\right\}.
\end{align*}
\end{thm}

On the right hand side of the above upper bound, the first term is from the tolerance parameter in the RE condition, which ensures that the condition can possibly hold when the true model is not exactly low-rank \citep{negahban2011}, i.e., when $\sum_{j=r_k+1}^{m_k}\sigma(\bB_{0k},j)\neq 0$. The second term gives the \textit{estimation error} of recovering the desired view-specific low-rank structure, and the third term gives the \textit{approximation error} incurred due to approximating the true model with the view-specific low-rank structure. When the true model is exactly of low rank, i.e., $r(B_{0k}) = r_{0k}$, it suffices to take $\delta=0$ and the upper bound then yields the estimation error, i.e., $\tau^2 \sum_{k=1}^{K}\{q + r(\bX_k)\}r_{0k}/n$. This rate holds with high probability in the high-dimensional setting that $q + r(\bX_k) \rightarrow \infty$. In the classical setting of  $n\rightarrow \infty$ with fixed $q$ and $r(\bX_k)$, by choosing $\theta \propto \sqrt{\log n}$, the rate becomes $\tau^2 \log(n)\sum_{k=1}^{K}r_{0k}/n$ with probability approaching 1. 

Intriguingly, the results in Theorem \ref{th:2} can specialize into oracle inequalities of several existing regularized estimation methods, such as NNP, MTL and Lasso. This is because these models can all be viewed as special cases of iRRR. As such, iRRR seamlessly bridges group-sparse and low-rank methods and provides a unified theory of the two types of regularization. Several examples are provided in Appendix \ref{sec:th:example}.

To see the potential advantage of iRRR over NNP or MTL, we make some comparisons of their error rates based on Theorem \ref{th:2}. To convey the main message, consider the case where $p_k = p_1$, $r(\X_k) = r_{X_1}$ for $k=1,\ldots, K$, $r_{0k} = r_{01}$ for $k=1,\ldots, s$, and $r_{0k} =0$ for $k=s+1,\ldots, K$. The error rate is $\tau^2sr_{01}(q + r_{X_1})/n$, $\tau^2r_0(q + r_X)/n$, for iRRR and NNP, respectively, with high probability. As long as $sr_{01} = O(r_0)$, iRRR achieves a faster rate since $r_{X_1} \leq r_X$ always holds. For comparing iRRR and MTL, we get that with probability $1-p^{-1}$, iRRR achieves an error rate $\tau^2(\log{p} + q + r_{X_1})sr_{01}/n$ (by choosing $\theta = \sqrt{4\log{p}/(q + r_{X_1})}$ ) while MTL achieves $\tau^2(\log{p}+q+1)sp_1/n$. The two rates agree with each other in the MTL setting when $r_{X_1} = r_{01} = p_1 =1$, and the former rate can be much faster in the iRRR setting when, for example, $r_{01} \ll p_1$ and $r_{X_1} = O(\log(p)+q)$. 

\section{Simulation}\label{sec:sim}


\subsection{Settings and Evaluation Metrics}

We conduct simulation studies to demonstrate the efficacy of the proposed iRRR method. We consider two response types: Gaussian and binary.
In Gaussian settings, we compare iRRR with the ordinary least squares (OLS), the ridge RRR (RRRR) \citep{mukherjee2011reduced} (which contains RRR as a special case), and the adaptive NNP (aNNP) \cite[which has been shown to be computationally efficient and can outperform NNP in][]{chen2012ann}.
For the settings in which the true coefficient matrix is sparse, we also include MTL \citep{Caruana1997} (by treating each predictor as a group in iRRR), as well as Lasso \citep{tibshirani1996regression} and grLasso \citep{yuan2006} for each response variable separately (grLasso accounts for the grouping information in the multi-view predictors). In binary settings, we compare iRRR with the generalized RRR (gRRR) \citep{she2012reduced,LuoLiang2017} and the univariate penalized logistic regression (glmnet) with the elastic net penalty \citep{zou2005}.

For the Gaussian models, we consider a range of simulation settings.
{\bf Setting 1} is the basic setting, where $n=500$, $K=2$, $p_1 = p_2 = 50$ ($p=100$), and $q = 100$. We generate the rows of the design matrix $\bX$ independently from a $p$-variate Gaussian distribution $\mbox{N}(\0,\bSigma_x)$ with $\bSigma_x = \I_p$, followed by column centering. The error matrix $\bE$ is filled with i.i.d.\ standard Gaussian random numbers. 
(We also consider correlated errors. The results are similar and contained in Appendix \ref{sec:appendix:sim}.)
Each coefficient matrix $\bB_{0k}$ has rank $r_{0k}=10$, which is generated as $\bB_{0k}=\bL_k\bR_k\trans$ with the entries of $\bL_k \in \mathbb{R}^{p_k\times r_{0k}}$ and $\bR_k\in \mathbb{R}^{q\times r_{0k}}$ both generated from $\mbox{N}(0,1)$. Consequently, $\bB_0=(\bB\trans_{01},\bB\trans_{02})\trans$ has rank $r_0=r_{01}+r_{02}=20$. The response matrix $\Y$ is then generated based on the model in \eqref{eq:grmodel}. As such, there are more than 10,000 unknown parameters in this model, posing a challenging large-scale problem.
Furthermore, we also consider incomplete responses, with 10\%, 20\%, 30\% entries missing completely at random.

The other settings are variants of {\bf Setting 1}:
\bi
\item {\bf Setting 2 (multi-collinear):} The predictors in the two views $\bX_1$ and $\bX_2$ are highly correlated. All the $p=p_1 + p_2$ predictors are generated jointly from a $p$-variate Gaussian distribution $\mbox{N}_{p}(\0, \bSigma_x)$, where $\bSigma_x$ has diagonal elements 1 and off-diagonal 0.9.
\item {\bf Setting 3 (globally low-rank):} We set $\bR_1=\bR_2$ when generating $\bB_{01}$ and $\bB_{02}$, so that the low rank structures in separate coefficient matrices also imply a globally low-rank structure. We consider three scenarios: $r_0=r_{01}=r_{02}=20$, $r_0=r_{01}=r_{02}=40$, and $r_0=60, r_{01}=r_{02}=50$. 
\item {\bf Setting 4 (multi-set):} We consider multiple views, $K \in \{ 3, 4,5\}$. The additional design matrices and coefficient matrices are generated in the same way as in {\bf Setting 1}.
\item {\bf Setting 5 (sparse-view):} We consider $K=3$, where the last predictor set $\X_3$ is generated in the same way as in {\bf Setting 1} but is irrelevant to prediction, i.e., $\bB_{03}=0$.
\ei

For the binary models, we consider two settings: the basic setting ({\bf Setting 6}) and the sparse-view setting ({\bf Setting 7}), which are similar to
{\bf Setting 1} and {\bf Setting 5}, respectively. The differences are that the sample size is set to $n=200$, the intercept $\bmu_0$ is set as a vector of random numbers from the uniform distribution on $[-1,1]$, and the entries of $\bY$ are drawn from Bernoulli distributions with their natural parameters given by $\bTheta = \1\bmu_0\trans+ \sum_{k=1}^{K}\bX_k\bB_{0k}$.

In {\bf Settings 1--5}, we use the MSPE to evaluate the performance of different methods,
\[
\mbox{MSPE}(\bB_0,\widehat{\bB})=\tr\left\{(\bB_0-\widehat{\bB})\trans\bSigma_{x}(\bB_0-\widehat{\bB})\right\},
\]
where $\tr(\cdot)$ represents the trace of a matrix, $\widehat{\bB}$ is the estimate of $\bB_0$, and $\bSigma_{x}$ is the covariance matrix of $\bX$.
In {\bf Settings 6--7}, we evaluate the average cross entropy between the true and estimated probabilities on an independently generated validation data set of size $n=500$,
\[
\mbox{En}(\bmu_0,\bB_0,\widehat{\bmu},\widehat{\bB})=-{1\over n}\sum_{i=1}^n\sum_{j=1}^q\left\{ p_{ij}\log\widehat{p}_{ij}+(1-p_{ij})\log(1-\widehat{p}_{ij})\right\},
\]
where $p_{ij}=\exp(\theta_{ij})/\{1+\exp(\theta_{ij})\}$, and $\widehat{p}_{ij}$ is its corresponding estimate. 

For each simulation setting,
we first generate an independent testing data set to select tuning parameters for different methods.
Once selected, the tuning parameters are fixed in subsequent analyses. This unified approach alleviates inaccuracy in the empirical tuning parameter selection to ensure a fair comparison of different regularization methods. We have also tried 5-fold CV. The results are similar to those from the validation data tuning and thus omitted for brevity. In each setting, the experiment is replicated 100 times.

\subsection{Results}

Table \ref{tab:normalpred} reports the results for {\bf Settings 1--4}. In all the settings, the three regularized estimation methods always substantially outperform OLS, indicating the strength and necessity of dimension reduction. In {\bf Setting 1 (basic)}, iRRR provides the best prediction performance, followed by aNNP and RRRR.
When the outcomes are incomplete, only iRRR is applicable.
The mean and standard deviation of MSPE over 100 repetitions are 7.87 (0.20), 8.64 (0.20), and 9.96 (0.24), when 10\%, 20\%, and 30\% of the responses are missing, respectively.
In {\bf Setting 2 (multi-collinear)}, iRRR is still the best. It is worth noting that owing to shrinkage estimation, RRRR slightly outperforms aNNP. In {\bf Setting 3 (globally low-rank)}, aNNP and RRRR can slightly outperform iRRR when $r_0$ is much smaller than $\sum_{k=1}^K r_{0k}$. This can be explained by the fact that under this setting iRRR may be less parsimonious than the globally reduced-rank methods. To see this, when $r_0$ is small and $r_0=r_{01}=r_{02}$, we have that $\sum_{k=1}^K(p_k+q-r_{0k})r_{0k}=\{p+K(q-r_0)\}r_0>(p+q-r_0)r_0$, i.e., iRRR yields a larger number of free parameters than RRR. Nevertheless, iRRR regains its superiority over the globally low-rank methods when $r_0$ becomes large. We remark that in multi-view problems the scenario of $r_0 \ll \sum_kr_{0k}$ rarely happens unless the relevant subspace from each view largely overlaps with each other. In {\bf Setting 4 (multi-set)}, we confirm that the advantage of iRRR becomes more obvious as the number of distinct view sets increases.


\begin{table}[htbp]
  \centering
  \caption{Simulation results for {\bf Settings 1--4}. The mean and standard deviation (in parenthesis) of MSPE over 100 simulation runs are presented. In each setting, the best results are highlighted in boldface.}\label{tab:normalpred}
  \begin{tabular}{lrrrr}
    \hline
     & \multicolumn{1}{c}{iRRR} & \multicolumn{1}{c}{aNNP} & \multicolumn{1}{c}{RRRR} & \multicolumn{1}{c}{OLS} \\\hline
    {\bf Setting 1} & {\bf 7.22} (0.17) & 7.76 (0.22) & 8.38 (0.24) & 25.15 (0.36) \\\hline
    {\bf Setting 2} & {\bf 4.21} (0.10) & 4.69 (0.11) & 4.52 (0.11) & 25.15 (0.36) \\\hline
    \hfill ($r_0=20$) & 10.13 (0.22) & {\bf 7.81} (0.25) & 8.25 (0.26) & 25.16 (0.39) \\
    {\bf Setting 3} ($r_0=40$) & 12.48 (0.19) & {\bf 12.39} (0.22) & 13.76 (0.26) & 25.04 (0.37) \\
    \hfill ($r_0=60$) & {\bf 13.62 }(0.21) & 14.66 (0.26) & 15.66 (0.17) & 25.11 (0.39) \\\hline
    \hfill($K=3$) & {\bf 10.19} (0.21) & 13.99 (0.32) & 15.44 (0.31) & 43.76 (0.59) \\
    {\bf Setting 4} \ ($K=4$) & {\bf 13.04} (0.22) & 19.99 (0.35) & 19.68 (0.19) & 68.00 (0.89) \\
    \hfill($K=5$) & {\bf 14.84} (0.25) & 24.90 (0.32) & 21.43 (0.21) & 101.87 (1.38) \\
    \hline
  \end{tabular}
\end{table}




Figure \ref{fig:sim5} displays the results for {\bf Setting 5 (sparse-view)}. We find that the iRRR solution tuned based on predictive accuracy usually estimates the third coefficient matrix (which is a zero matrix in truth) as a nearly zero matrix and occasionally an exact zero matrix; in view of the construction of the cNNN penalty in iRRR, this ``over-selection'' property is analogous to that of Lasso or grLasso. Motivated by \cite{zou2006adaptive}, we also experiment with an adaptive iRRR (denoted by iRRR-a) approach, where we first fit iRRR and then adjust the predefined weights by the inverse of the Frobenius norms of the estimated coefficient matrices. As a result, the iRRR-a approach achieves much improved view selection performance and even better prediction accuracy than iRRR. In contrast, MTL, Lasso and grLasso have worse performance than the low-rank methods, because they fail to leverage information from the multivariate response and/or multi-view predictor structures.

\begin{figure}[ht]
  \centering
  \includegraphics[width=4in]{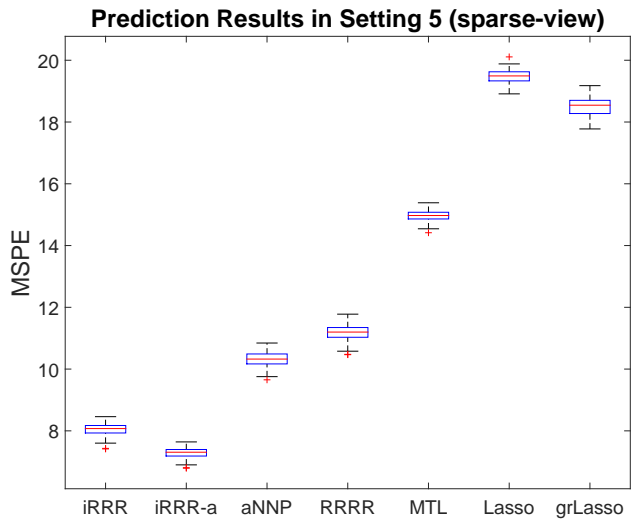}
  \caption{Simulation results for {\bf Setting 5 (sparse-view)}. OLS is omitted as its performance is much worse than the reported methods.}\label{fig:sim5}
\end{figure}


The simulation results of {\bf Settings 6-7} for binary models are displayed in Figure \ref{fig:sim6}. The results are similar as in the Gaussian models, i.e., the iRRR methods substantially outperform the competing sparse or low-rank methods in prediction. 


\begin{figure}[htp]
  \centering
  \includegraphics[width=3in]{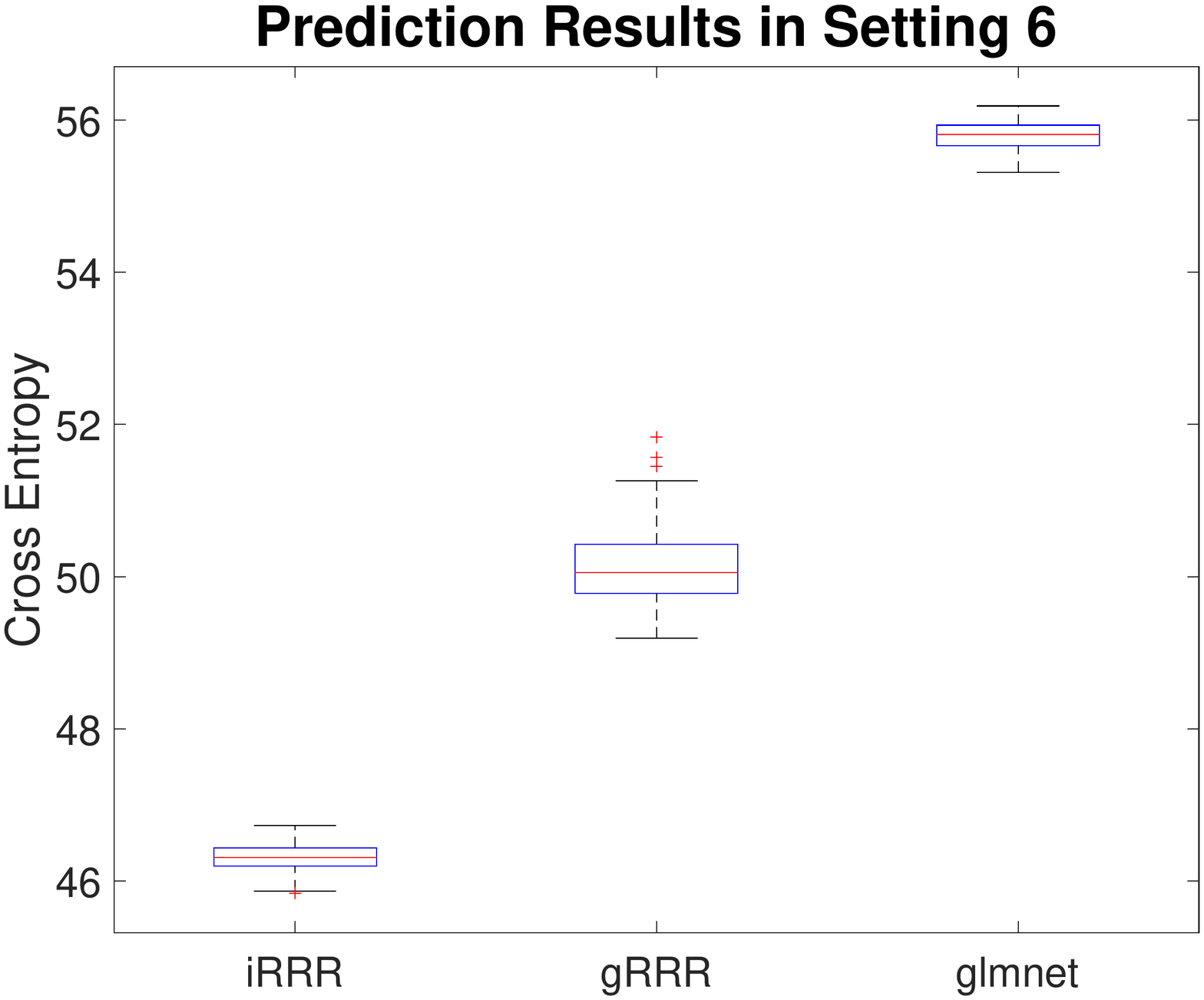}
  \includegraphics[width=3in]{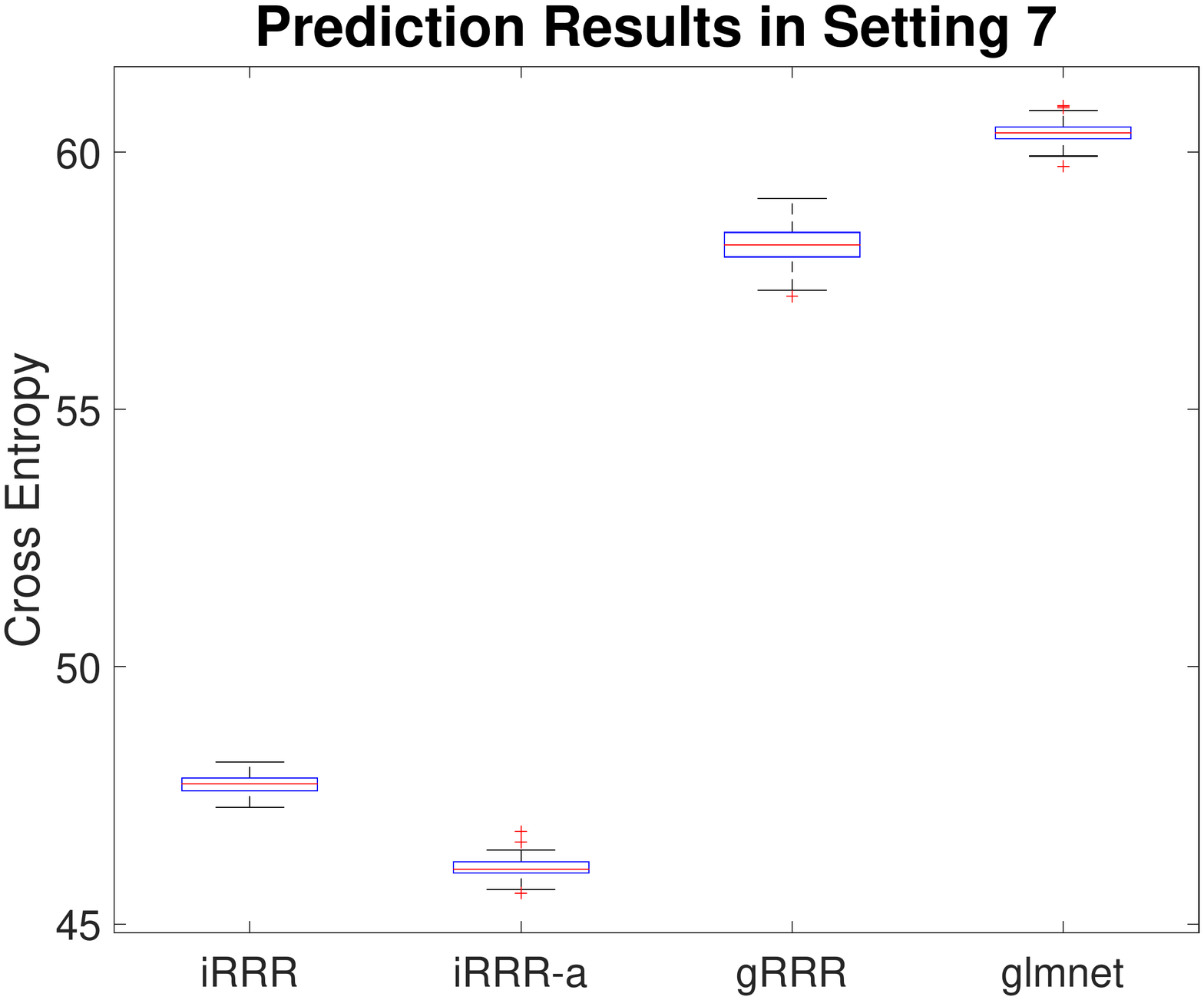}
    \caption{Simulation results for {\bf Settings 6--7} with binary response variables.}\label{fig:sim6}
\end{figure}
%

We have also compared the computational time of different methods (on a standard desktop with Intel i5 3.3GHz CPU). For example, the average time (in seconds) under {\bf Setting 1} is  $0.68$ $(0.06)$, $0.07$ $(0.01)$ and  $0.02$ $(0.00)$ for iRRR, aNNP and RRRR, respectively; under {\bf Setting 4} with $K=5$ the average time becomes $0.96$ $(0.12)$, $0.09$ $(0.01)$ and $0.05$ $(0.01)$; under {\bf Setting 6} with binary responses, the average time is $1.71$ $(0.03)$, $0.98$ $(0.08)$ and $0.70$ $(0.08)$ for iRRR, gRRR and glmnet. As expected, iRRR is more computationally expensive than the globally low-rank or sparse methods. However, in view of the scale of the problem, the computational cost for iRRR is still low and acceptable. The Matlab code for implementing the proposed method is available at \href{https://github.com/reagan0323/iRRR}{https://github.com/reagan0323/iRRR}.


\section{An Application in the Longitudinal Studies of Aging}\label{sec:real}


The LSOA \citep{stanziano2010review} was a collaborative effort of the National Center for Health Statistics and the National Institute on Aging. The study interviewed a large cohort of senior people (70 years of age and over) in 1997-1998 (WAVE II) and 1999-2000 (WAVE III), respectively, and measured their health conditions, living conditions, family situations, health service utilizations, among others. Here our objective is to examine the predictive relationship between health-related events in earlier years and health outcomes in later years, which can be formulated as a multivariate regression problem.

There are $n=3988$ common subjects who participated in both WAVE II and WAVE III interviews. After data pre-processing \citep{LuoLiang2017}, $p=294$ health risk and behavior measurements in WAVE II are treated as predictors, and $q=41$ health outcomes in WAVE III are treated as multivariate responses. The response variables are binary indicators, characterizing various cognitive, sensational, social, and life quality outcomes, among others. Over 20\% of the response data entries are missing. The predictors are multi-view, including housing condition ($\bX_1$ with $p_1=38$), family structure/status ($\bX_2$ with $p_2=60$), daily activity ($\bX_3$ with $p_3=40$), prior medical condition ($\bX_4$ with $p_4=114$), and medical procedure since last interview ($\bX_5$ with $p_5=40$). We thus apply the proposed iRRR method to perform the regression analysis. As a comparison, we also implement gRRR \citep{LuoLiang2017}, and both classical and sparse logistic regression methods using the R package \texttt{glmnet}, denoted as glm and glmnet, respectively.




We use a random-splitting procedure to evaluate the performance of different methods. More specifically, each time we randomly select $n_{tr}=3000$ subjects as training samples and the remaining $n_{te}=988$ subjects as testing samples. For each method, we use 5-fold CV on the training samples to select tuning parameters, and apply the method to all the training data with the selected tuning parameters to yield its coefficient estimate. The performance of each method is measured by the average deviance between the observed true response values and the estimated probabilities, defined as
\[
\mbox{Average Deviance} ={-2\sum_{i=1}^{n_{te}}\sum_{j=1}^q\{y_{ij}\log\widehat{p}_{ij}+(1-y_{ij})\log(1-\widehat{p}_{ij})\}\delta_{ij}\over \sum_{i=1}^{n_{te}}\sum_{j=1}^q\delta_{ij}},
\]
where $\delta_{ij}$ is an indicator of whether $y_{ij}$ is observed.
We also calculate the Area Under the Curve (AUC) of the Receiver Operating Characteristic (ROC) curve for each outcome variable. This procedure is repeated 100 times and the results are averaged.

In terms of the average deviance, iRRR and glmnet yield very similar results (with mean 0.77 and standard deviation 0.01), and both substantially outperform gRRR (with mean 0.83 and standard deviation 0.01) and glm (fails due to a few singular outcomes). The out-sample AUCs for different response variables are shown in Figure \ref{fig:LSOA}. The response variables are sorted based on their missing rates from large (over 70\%) to small (about 13\%). Again, the performance of iRRR is comparable to that of glmnet. The iRRR tends to have a slight advantage over glmnet for responses with high missing rates. This could be due to the fact that iRRR can borrow information from other responses while the univariate glmnet cannot.

\begin{figure}[ht!]
  \centering
  \includegraphics[width=4in]{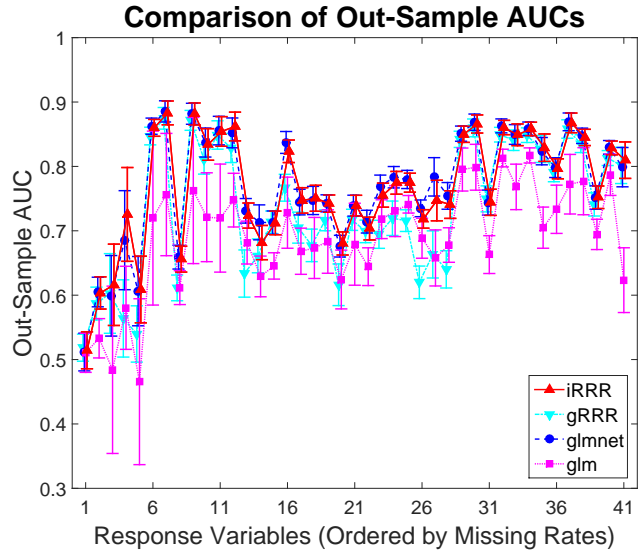}
  \caption{LSOA data analysis. The mean and standard deviation (error bar) of AUC for each response variable over 100 random-splitting procedures. The responses, from left to right, are ordered by missing rates from large to small.}\label{fig:LSOA}
\end{figure}



To understand the impact of different views on prediction, we produce heatmaps of the estimated coefficient matrices in Figure \ref{fig:heat} (glm is omitted due to its poor performance).
The estimates from iRRR and glmnet show quite similar patterns: it appears that the family structure/status group and the daily activity group have the most predictive power, and the variables within these two groups contribute to the prediction in a collective way. As for the other three views, iRRR yields heavily shrunk coefficient estimates, while glmnet yields very sparse estimates. These agreements partly explain the similarity of the two methods in their prediction performance. In contrast, the gRRR method tries to learn a globally low-rank structure rather than a view-specific structure; consequently, it yields a less parsimonious solution with less competitive prediction performance. Therefore, our results indicate that generally knowing the family structure/status and daily activity measurements, the information on housing condition, prior medical conditions, and medical procedures do not provide much new contribution to the prediction of health outcomes on cognition, sensation, social behavior, life quality, among others.



\begin{figure}[ht!]
  \centering
  \includegraphics[width=6in]{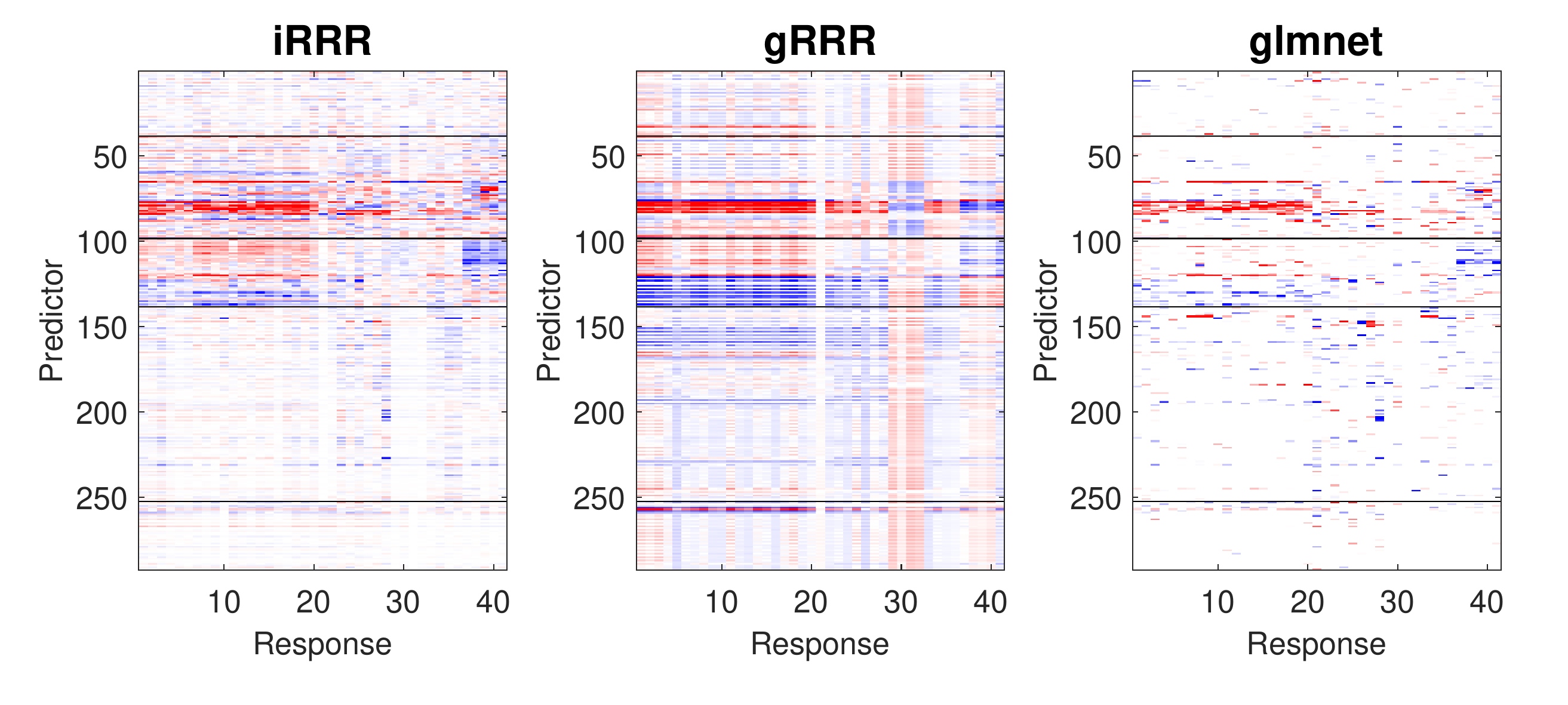}
  \caption{LSOA data analysis. The heat maps of the coefficient matrices estimated from different methods. The predictors fall into 5 groups, namely, housing condition, family status, daily activity, prior medical condition, and change in medical procedure since last interview, from top to bottom separated by horizontal black lines. For visualization purpose, we also sort the responses based on their grouping structure (e.g., cognition, sensation, social behavior, and life quality). }\label{fig:heat}
\end{figure}



\section{Discussion}\label{sec:dis}
With multi-view predictor/feature sets, it is likely that some of the views are irrelevant to the prediction of the outcomes, and the features within a relevant view may be highly correlated and hence contribute to the prediction collectively rather than sparsely. When dealing with such problem, the two commonly used methodologies, i.e., sparse methods and low-rank methods, both have shortcomings. The joint extraction of latent features from each view in a supervised fashion offers a better solution; indeed, this is what iRRR strives to achieve.

There are many directions for future research. For conducting simultaneous view selection and within-view subspace selection, the proposed cNNP scheme can be extended to a general \emph{composite singular value penalization} scheme,
$\lambda \sum_{k=1}^{K} w_k \rho_{\mathcal{O}}\left( \sum_{j=1}^{p_k\wedge q} \rho_{\mathcal{I}}\left( \sigma(\bB_{k},j) \right) \right)$,
where $\rho_{\mathcal{I}}$ is an inner penalty function for inducing sparsity among the singular values of each $\bB_{k}$, and $\rho_{\mathcal{O}}$ is an outer penalty function for enforcing sparsity among the $\bB_{k}$ matrices. For example, the family of bridge penalties \citep{huang2008} can be used in both inner and outer penalization. Incorporating sparse within-view variable selection to iRRR could also be fruitful; one way to achieve this is to use an additive penalty form of cNNP and grLasso. Moreover, it is possible to combine iRRR with a covariate-adjusted (inverse) covariance estimation method \citep{rothman2010sparse}, to jointly estimate the mean and covariance structures. Another pressing problem is to generalize iRRR to handle heterogeneous data, as in practice data may be count-valued, interval-valued, or mixed of several types with substantial missing values \citep{LuoLiang2017}. Computationally, the ADMM algorithm can be coupled with a Majorization-Minimization algorithm to handle these cases.

\if1\blind
{
\section*{Acknowledgement}

Gen Li's research was partially supported by Calderone Junior Faculty Award from the Mailman School of Public Health at Columbia University. Kun Chen's research is partially supported by National Science Foundation grants DMS-1613295 and IIS-1718798, and National Institutes of Health grants U01-HL114494 and R01-MH112148. 
}\fi


\appendix 

\section{Details on Theoretical Analysis}\label{sec:appendix:th}
\subsection{On the Restricted Eigenvalue Condition}\label{sec:th:rec}

To specify the restricted set $\mathcal{C}$, we need some additional constructions. For each $\bB_{0k} \in \mathbb{R}^{p_k\times q}$ ($k=1,\ldots,K$), let $\bB_{0k} = \bU_k\bD_k\bV_k\trans$ be its full SVD, where $\bU_k \in\mathbb{R}^{p_k\times p_k}$, $\bV_k \in \mathbb{R}^{q\times q}$ satisfy $\bU_k\trans\bU_k = \I_{p_k}$ and $\bV_k\trans\bV_k = \bI_q$. For each $r \in \{1,2,\cdots,m_k\}$, where $m_k=p_k\wedge q$, let $\bU_k^r$, $\bV_k^r$ be the submatrices of singular vectors associated with the top $r$ singular values of $\bB_{0k}$. Define the following subspaces of $\mathbb{R}^{p_k\times q}$:
\begin{eqnarray*}
\mathcal{A}(\bU_k^r,\bV_k^r) &=& \{\Delta_k \in \mathbb{R}^{p_k\times q}; \mbox{row}(\Delta_k) \subset \bV_k^r, \mbox{col}(\Delta_k)\subset \bU_k^r\},\\
\mathcal{B}(\bU_k^r,\bV_k^r) &=& \{\Delta_k \in \mathcal{R}^{p_k\times q}; \mbox{row}(\Delta_k) \perp \bV_k^r, \mbox{col}(\Delta_k)\perp \bU_k^r\},
\end{eqnarray*}
where $\mbox{row}(\Delta_k)$ and $\mbox{col}(\Delta_k)$ denote the row space and column space of $\Delta_k$, respectively. We may adopt the shorthand notation $\mathcal{A}_k^r$ and $\mathcal{B}_k^r$ when no confusion arises. Let $\mathcal{P}_{\mathcal{B}_k^{r_k}}$ denote the projection operator onto the subspace $\mathcal{B}_k^{r_k}$, and define $\Delta_k'' = \mathcal{P}_{\mathcal{B}_k^{r_k}}(\Delta_k)$ and $\Delta_k'=\Delta_k - \Delta_k''$. We now define the restricted set
\begin{align}
& \mathcal{C}(r_1,\ldots, r_K;\delta)\notag\\
 = &  \left\{
 \Delta \in \mathbb{R}^{p \times q}; \|\Delta\|_\Fn \geq \delta, \sum_{k=1}^{K}w_k\|\Delta_k''\|_\star \leq \sum_{k=1}^{K}\{ 3w_k\|\Delta_k'\|_\star + 4w_k \sum_{j=r_k+1}^{m_k}\sigma(\bB_{0k},j)\}
\right\}.\label{eq:restrictedset}
\end{align}
where $\delta$ is a tolerance parameter and $w_k=\sigma(\bX_k,1)\{\sqrt{q} + \sqrt{r(\bX_k)}\}/n$, as defined in
\eqref{eq:weight}. We refer to \citet{negahban2011} and \citet{negahban2012} for examples of the restricted set, including the cases of Lasso, grLasso and NNP.


\subsection{Special Cases of Theorem \ref{th:2}}\label{sec:th:example}

The results on iRRR in Theorem \ref{th:2} can specialize into oracle inequalities of several existing regularized estimation methods, such as NNP, MTL and Lasso. We discuss some examples below; to focus on the main message, we only focus on the settings of exact low rank or exact sparsity. First consider the NNP method defined in \eqref{eq:NNP}, which corresponds to the special case of $K=1$ and $w_k=1$ in iRRR. The restricted set in \eqref{eq:restrictedset} becomes
\begin{align*}
\mathcal{C}(r_0)=\{\Delta \in \mathbb{R}^{p\times q}; \|\Delta''\|_\star \leq 3\|\Delta'\|_\star\},
\end{align*}
where $\Delta'' = \mathcal{P}_{\mathcal{B}_0^{r_0}}(\Delta)$ and $\Delta'=\Delta - \Delta''$. Theorem \ref{th:2} then implies that under the RE condition with $\kappa(\bX)>0$ over $\mathcal{C}(r_0)$, if we choose $\lambda=2\tau(1+\theta) \sigma(\bX,1)\{\sqrt{q}+\sqrt{r(\bX)}\}$, then with probability at least
$1 - \exp[-\theta^2\{q+r(\bX)\}/2]$, it holds that
\begin{align*}
\|\widehat{\bB} - \bB_0\|_\Fn^2 \preceq \frac{\tau^2}{\kappa(\bX)^2}\frac{\{\sqrt{q}+\sqrt{r(\bX)}\}^2 r_0}{n}.
\end{align*}
This bound recovers the results on NNP in the literature; see, e.g., \citet{negahban2011}. Next, consider the MTL setting, which corresponds to $p_k=1$ and $p=K$ in iRRR. Write $\bB_0 = (\bb_{01}\trans,\ldots, \bb_{0p}\trans)\trans \in \mathbb{R}^{p\times q}$, and $\mathcal{S} = \{j; \|\bb_{0j}\|_2 \neq 0\}$. The restricted set becomes
$$
\mathcal{C}(\mathcal{S})=\left\{\Delta= (\Delta_1\trans,\ldots,\Delta_p\trans)\trans \in \mathbb{R}^{p\times q}; \sum_{k \in S^c}\|\Delta_k\|_2 \leq 3\sum_{k \in S}\|\Delta_k\|_2\right\}.
$$
By choosing 
$\lambda \propto \tau \sqrt{\log p /q}$, Theorem \ref{th:2} yields the high probability bound
\begin{align*}
\|\widehat{\bB} - \bB_0\|_\Fn^2 \preceq \frac{\tau^2}{\kappa(\bX)^2}\frac{(\log p +q)\cdot |\mathcal{S}|}{n},
\end{align*}
where $|\mathcal{S}|$ is the cardinality of $\mathcal{S}$. The same bound can be obtained from results in \citet{lounici2011} on more general setting of MTL, or from results in \citet{negahban2012} on grLasso by vectorizing the MTL problem here into a univariate-response regression. Another example is Lasso, which corresponds to $q=1$ and $K=p$ in iRRR. It is seen that the model becomes $\by = \bX\bb_{0} + \be$, and the cNNP degenerates to the $\ell_1$-norm of a coefficient vector $\bb \in \mathbb{R}^p$. Let $\mathcal{S} = \{j; b_{0j}\neq 0\}$, then the restricted set becomes
$$
\mathcal{C}(\mathcal{S})=\left\{\Delta = (\Delta_1,\ldots,\Delta_p)\trans \in \mathbb{R}^p; \sum_{k \in S^c}|\Delta_k| \leq 3\sum_{k \in S}|\Delta_k|\right\}.
$$
Theorem \ref{th:2} implies that by choosing $\lambda \propto \tau \sqrt{c\log{p}}$,
$$
\|\widehat{\bb} - \bb_0\|_2^2 \preceq \frac{\tau^2 }{\kappa(\bX)^2}\frac{\log{p}\cdot|\mathcal{S}|}{n}
$$
holds with probability at least $1-p^{1-c}$,  which is a well-known result in the literature.


\subsection{Proof of Theorem \ref{th:1} and a Corollary on the Estimation Error}\label{sec:supp:th1}

\begin{proof}[Proof of Theorem \ref{th:1}]
By definition,
\begin{align*}
& \|\bY - \bX\widehat{\bB}\|_\Fn^2 + 2\lambda \sum_{k=1}^{K}\sigma(\bX_k,1)(\sqrt{q} + \sqrt{r(\bX_k)})\|\bB_k\|_\star \\
\leq & \|\bY - \bX\bC\|_\Fn^2 + 2\lambda \sum_{k=1}^{K}\sigma(\bX_k,1)(\sqrt{q} + \sqrt{r(\bX_k)})\|\bC_k\|_\star,
\end{align*}
which leads to
\begin{align*}
\|\bX\widehat{\bB} - \bX\bB_0\|_\Fn^2
\leq & \|\bX\bC - \bX\bB_0\|_\Fn^2
+ 2\lambda \sum_{k=1}^{K}\sigma(\bX_k,1)(\sqrt{q} + \sqrt{r(\bX_k)})\|\bC_k\|_\star \\
& + 2\langle \bX\trans\bE,\widehat{\bB}-\bC \rangle_\Fn  - 2\lambda \sum_{k=1}^{K}\sigma(\bX_k,1)(\sqrt{q} + \sqrt{r(\bX_k)})\|\widehat{\bB}_k\|_\star.
\end{align*}


Define an event $\mathcal{A}_k = \{\sigma(\bX_k\trans\bE,1)\leq \lambda \sigma(\bX_k,1)(\sqrt{q}+\sqrt{r(\bX_k)})\}$, for $k=1,\ldots, K$. First, consider the inner product term. On the event $\cap_{k=1}^{K}\mathcal{A}_k$, we have
\begin{align*}
 \langle \bX\trans\bE,\widehat{\bB}-\bC \rangle_\Fn
= & \mbox{tr}\{\bE\trans\bX(\widehat{\bB}-\bC)\}\\
= & \sum_{k=1}^{K}\langle \bX_k\trans\bE,\widehat{\bB}_k - \bC_k \rangle_\Fn\\
\leq & \sum_{k=1}^{K}\sigma(\bX_k\trans\bE,1)\|\widehat{\bB}_k-\bC_k\|_\star\\
\leq & \lambda \sum_{k=1}^{K}\sigma(\bX_k,1)(\sqrt{q}+\sqrt{r(\bX_k)})\|\widehat{\bB}_k-\bC_k\|_\star.
\end{align*}
It follows that on the event $\cap_{k=1}^{K}\mathcal{A}_k$,
\begin{align}
\|\bX\widehat{\bB} - \bX\bB_0\|_\Fn^2
\leq & \|\bX\bC - \bX\bB_0\|_\Fn^2
+ 2\lambda \sum_{k=1}^{K}\sigma(\bX_k,1)(\sqrt{q} + \sqrt{r(\bX_k)})\|\bC_k\|_\star \notag\\
& + 2\lambda \sum_{k=1}^{K}\sigma(\bX_k,1)(\sqrt{q}+\sqrt{r(\bX_k)})\|\widehat{\bB}_k-\bC_k\|_\star\notag\\
& - 2\lambda \sum_{k=1}^{K}\sigma(\bX_k,1)(\sqrt{q} + \sqrt{r(\bX_k)})\|\widehat{\bB}_k\|_\star\notag\\
\leq & \|\bX\bC - \bX\bB_0\|_\Fn^2 + 4\lambda \sum_{k=1}^{K}\sigma(\bX_k,1)(\sqrt{q}+\sqrt{r(\bX_k)})\|\bC_k\|_\star,\label{eq:bound1}
\end{align}
where the last inequality is due to the triangle inequality.

Now we consider the probability of the event $\cap_{k=1}^{K}\mathcal{A}_k$. Let $\mathcal{P}$ be the projection matrix onto the column space of $\bX$, and $\mathcal{P}_k$ be the projection matrix onto the column space of $\bX_k$, for $k=1,\ldots, K$. Because $\sigma(\bX_k\trans\bE,1) \leq \sigma(\bX_k,1)\sigma(\mathcal{P}_k\bE,1)$,
we have
\begin{align*}
\cap_{k=1}^{K}\mathcal{A}_k = & \{\sigma(\bX_k\trans\bE,1)\leq \lambda \sigma(\bX_k,1)(\sqrt{q}+\sqrt{r(\bX_k)});k=1,\ldots,K\}\\
\supseteq & \{\sigma(\mathcal{P}_k\bE,1)\leq \lambda(\sqrt{q}+\sqrt{r(\bX_k)});k=1,\ldots, K \}\\
\equiv & \cap_{k=1}^{K}\tilde{\mathcal{A}}_k.
\end{align*}
By Lemma 3 in \citet{bunea2011optimal},
$$
\mathbb{P}\{(\sigma(\mathcal{P}_k\bE,1)\geq \mathbb{E}[\sigma(\mathcal{P}_k\bE,1)+\tau t]\} \leq \exp(-t^2/2),
$$
and $\mathbb{E}[\sigma(\mathcal{P}_k\bE,1)] \leq \tau (\sqrt{q} + \sqrt{r(\bX_k)})$, for any $k=1,\ldots,K$. Therefore,
$$
\mathbb{P}\{\cup_{k=1}^{K} \tilde{\mathcal{A}}_k^c\} \leq \sum_{k=1}^{K}\exp\{-\frac{1}{2}\theta^2(q+r(\bX_k))\}.
$$
It then follows that
\begin{align*}
\mathbb{P}\{\cap_{k=1}^{K}\mathcal{A}_k\}
\geq \mathbb{P}\{\cap_{k=1}^{K}\tilde{\mathcal{A}}_k\}
= 1 - \mathbb{P}\{\cup_{k=1}^{K}\tilde{\mathcal{A}}_k^c\}
\geq 1 - \sum_{k=1}^{K}\exp\{-\frac{1}{2}\theta^2(q+r(\bX_k))\}.
\end{align*}
This, together with \eqref{eq:bound1}, completes the proof.
\end{proof}

\begin{cor*}
Assume that $\bE$ has i.i.d.\ $\mbox{N}(0,\tau^2)$ entries, and assume  $\sigma(\bX, p) > 0$. Let $\lambda = (1+\theta)\tau$, with $\theta >0$ arbitrary. Then with probability at least
$1 - \sum_{k=1}^{K}\exp[-\theta^2\{q+r(\bX_k)\}/2]$,
\begin{align*}
\|\widehat{\bB} - \bB_0\|_\Fn^2 &
\preceq
\tau^2 (1+\theta)^2\sum_{k=1}^{K}\frac{\Lambda(\bZ_k,1)}{\Lambda(\bZ,p)^2} \frac{\{\sqrt{q} + \sqrt{r(\bX_k)}\}^2r_{0k}}{n},
\end{align*}
where ``$\preceq$'' means the inequality holds up to some multiplicative constant.
\end{cor*}


The corollary shows that the estimation error rate for iRRR is $\tau^2 \sum_{k=1}^{K}\{q + r(\bX_k)\}r_{0k}/n$. This is potentially better than $\tau^2\{q + r(\bX)\}r_0/n$, the rate achieved by the NNP estimator under the same conditions \citep{bunea2011optimal}; for example, when $r(\bX) = \sum_kr(\bX_k)$ and $r(\bB_{0}) = \sum_k r_{0k}$.

\begin{proof}[Proof of Corollary] 
From the proof of Theorem \ref{th:1},
\begin{align*}
\|\bX\widehat{\bB} - \bX\bB_0\|_\Fn^2 \leq & \|\bX\bC - \bX\bB_0\|_\Fn^2 \\
& +
2\lambda \sum_{k=1}^{K}\sigma(\bX_k,1)(\sqrt{q}+\sqrt{r(\bX_k)})\{\|\bC_k\|_\star + \|\widehat{\bB}_k - \bC_k\|_\star - \|\widehat{\bB}_k\|_\star\}.
\end{align*}
With the results in the proof of their Theorem 12 in \citet{bunea2011optimal}, we have
\begin{align*}
& \|\bX\widehat{\bB} - \bX\bB_0\|_\Fn^2
- \|\bX\bC - \bX\bB_0\|_\Fn^2 \\
\leq &
4\lambda \sum_{k=1}^{K}\sigma(\bX_k,1)(\sqrt{q}+\sqrt{r(\bX_k)})\sqrt{3 r(\bC_k)}\|\widehat{\bB}_k - \bC_k\|_\Fn\\
\leq & 4 \lambda \sqrt{3\sum_{k=1}^{K}\sigma(\bX_k,1)^2(\sqrt{q} + \sqrt{r(\bX_k)})^2r(\bC_k)} \sqrt{\sum_{k=1}^K\|\widehat{\bB}_k-\bC_k\|_\Fn^2}\\
\leq &
\frac{4\sqrt{3}\lambda }{\sigma(\bX,p)}
\sqrt{\sum_{k=1}^{K}\sigma(\bX_k,1)^2(\sqrt{q} + \sqrt{r(\bX_k)})^2r(\bC_k)}\|\bX\widehat{\bB} - \bX\bC\|_\Fn\\
\leq & \frac{1}{2}\|\bX\widehat{\bB} - \bX\bC\|_\Fn^2 +  \frac{24\lambda^2}{\sigma(\bX,p)^2}(\sum_{k=1}^{K}\sigma(\bX_k,1)^2(\sqrt{q} + \sqrt{r(\bX_k)})^2r(\bC_k)).
\end{align*}
It follows that
\begin{align*}
\|\bX\widehat{\bB} - \bX\bB_0\|_\Fn^2
\leq &  3\|\bX\bC - \bX\bB_0\|_\Fn^2
 + \frac{48\lambda^2}{\sigma(\bX,p)^2}(\sum_{k=1}^{K}\sigma(\bX_k,1)^2(\sqrt{q} + \sqrt{r(\bX_k)})^2r(\bC_k)).
\end{align*}
Taking $\bC = \bB_0$ leads to the bound for the prediction error
\begin{align*}
\|\bX\widehat{\bB} - \bX\bB_0\|_\Fn^2 &
\preceq
\tau^2 \sum_{k=1}^{K}\frac{\Lambda(\bZ_k,1)}{\Lambda(\bZ,p)} (\sqrt{q} + \sqrt{r(\bX_k)})^2r_{0k}.
\end{align*}
Then by using the fact that $\sigma(\bX,p) > 0$ we get the claimed bound.

\end{proof}

\subsection{Proof of Theorem \ref{th:2}}\label{sec:supp:th2}

\begin{proof}[Proof of Theorem \ref{th:2}]
By definition, we have
\begin{align*}
\frac{1}{2n}\|\bY-\bX\widehat{\bB}\|_\Fn^2+\lambda\sum_{k=1}^{K}w_k\|\widehat{\bB}_k\|_\star \leq \frac{1}{2n}\|\bY-\bX\bB_0\|_\Fn^2+\lambda\sum_{k=1}^{K}w_k\|\bB_{0k}\|_\star,
\end{align*}
which leads to
\begin{align}
\frac{1}{2n}\|\bX\Delta\|_\Fn^2 \leq \lambda\sum_{k=1}^{K}w_k(\|\bB_{0k}\|_\star-\|\widehat{\bB}_k\|_\star)+\frac{1}{n}\langle \bE, \bX\Delta \rangle_\Fn,\label{eq:app1}
\end{align}
where $\Delta=\widehat{\bB}-\bB_0$.

Firstly, we verify that $\Delta$ belongs to the restricted set defined in \eqref{eq:restrictedset} so that the RE condition can be applied. Consider the first term on the right hand side of \eqref{eq:app1}. With the projection operators defined in Appendix \ref{sec:th:rec}, we have that
\begin{align*}
 \|\mathcal{P}_{\mathcal{A}_k^{r_k}}(\bB_{0k})+\Delta_k''\|_\star = \|\mathcal{P}_{\mathcal{A}_k^{r_k}}(\bB_{0k})\|_\star+\|\Delta_k''\|_\star,
 \end{align*}
and
 \begin{align*}
 \|\widehat{\bB}_k\|_\star & = \|\mathcal{P}_{\mathcal{A}_k^{r_k}}(\bB_{0k})+\Delta_k''+\mathcal{P}_{\mathcal{B}_k^{r_k}}(\bB_{0k})+\Delta_k'\| \\
& \geq \|\mathcal{P}_{\mathcal{A}_k^{r_k}}(\bB_{0k})+\Delta_k''\|_\star-\|\mathcal{P}_{\mathcal{B}_k^{r_k}}(\bB_{0k})+\Delta_k'\|_\star \\
& = \|\mathcal{P}_{\mathcal{A}_k^{r_k}}(\bB_{0k})\|_\star+\|\Delta_k''\|_\star-\|\mathcal{P}_{\mathcal{B}_k^{r_k}}(\bB_{0k})\|_\star-\|\Delta_k'\|_\star.
\end{align*}
Therefore,
\begin{align}
\nonumber \|\bB_{0k}\|_\star-\|\widehat{\bB}_k\|_\star & \leq \|\mathcal{P}_{\mathcal{A}_k^{r_k}}(\bB_{0k})\|_\star+\|\mathcal{P}_{\mathcal{B}_k^{r_k}}(\bB_{0k})\|_\star \\
\nonumber & - (\|\mathcal{P}_{\mathcal{A}_k^{r_k}}(\bB_{0k})\|_\star+\|\Delta_k''\|_\star-\|\mathcal{P}_{\mathcal{B}_k^{r_k}}(\bB_{0k})\|_\star-\|\Delta_k'\|_\star)\\
& = 2\|\mathcal{P}_{\mathcal{B}_k^{r_k}}(\bB_{0k})\|_\star+\|\Delta_k'\|_\star-\|\Delta_k''\|_\star.\label{eq:bound2}
\end{align}
We then deal with the second term on the right hand side of \eqref{eq:app1}. We have
\begin{align}
\nonumber \langle \bE, \bX\Delta \rangle_\Fn & = \mbox{tr}(\bE\trans\bX\Delta)\\
\nonumber & = \sum_{k=1}^{K} \langle \bX_k\trans\bE, \Delta_k \rangle_\Fn \\ & \leq \sum_{k=1}^{K}\sigma(\bX_k\trans\bE,1)\|\Delta_k\|_\star. \label{eq:bound3}
\end{align}
Combining results in \eqref{eq:bound2} and \eqref{eq:bound3}, we get
\begin{align*}
\frac{1}{2n}\|\bX\Delta\|_\Fn^2 \leq \lambda\sum_{k=1}^{K}w_k(2\|\mathcal{P}_{\mathcal{B}_k^{r_k}}(\bB_{0k})\|_\star+\|\Delta_k'\|_\star-\|\Delta_k''\|_\star)+\frac{1}{n}\sum_{k=1}^{K}\sigma(\bX_k\trans\bE,1)\|\Delta_k\|_\star.
\end{align*}
Define an event $\mathcal{A}_k = \{\sigma(\bX_k\trans\bE,1)/n \leq \lambda w_k /(1+ \eta)\}$, for $k=1,\ldots, K$, where $\eta > 0$ is an arbitrary positive number. It follows that on the event $\cap_{k=1}^{K}\mathcal{A}_k$,
\begin{align*}
0 \leq \frac{1}{2n}\|\bX\Delta\|_\Fn^2 & \leq \lambda\sum_{k=1}^{K}w_k(2\|\mathcal{P}_{\mathcal{B}_k^{r_k}}(\bB_{0k})\|_\star+\|\Delta_k'\|_\star-\|\Delta_k''\|_\star) \\
& +\frac{\lambda}{1+\eta}\sum_{k=1}^{K}w_k\|\Delta_k\|_\star\\
& \leq \lambda\sum_{k=1}^{K}w_k(2\|\mathcal{P}_{\mathcal{B}_k^{r_k}}(\bB_{0k})\|_\star+\frac{2+\eta}{1+\eta}\|\Delta_k'\|_\star-\frac{\eta}{1+\eta}\|\Delta_k''\|_\star)
\end{align*}
Therefore, it holds that
\begin{align}
\sum_{k=1}^{K}w_k\|\Delta_k''\|_\star \leq \frac{2+2\eta}{\eta}\sum_{k=1}^{K}w_k\sum_{j=r_k+1}^{m_k}\sigma(\bB_{0k},j)
 +\frac{2+\eta}{\eta}\sum_{k=1}^{K}w_k\|\Delta_k'\|_\star.\label{eq:ap3}
\end{align}
Taking $\eta = 1$ and assuming $\|\Delta\|_\Fn \geq \delta$, we see that $\Delta \in \mathcal{C}(r_1,\cdots, r_k,\delta)$. Therefore, based on the RE condition,
\begin{align}
\kappa(\bX)\|\Delta\|_\Fn^2 \leq \frac{1}{2n}\|\bX\Delta\|_\Fn^2.\label{eq:ap2}
\end{align}


From \eqref{eq:bound3} and on the event $\cap_{k=1}^{K}\mathcal{A}_k$, we have
\begin{align}
\frac{1}{2n}\|\bX\Delta\|_\Fn^2 & \leq \lambda\sum_{k=1}^{K}w_k(\|\bB_{0k}\|_\star-\|\widehat{\bB}_k\|_\star)+\frac{1}{n} \langle \bE, \bX\Delta \rangle_\Fn \notag\\
& \leq \lambda \sum_{k=1}^{K}w_k\|\Delta_k\|_\star+\frac{\lambda}{1+\eta}\sum_{k=1}^{K}w_k\|\Delta_k\|_\star \notag\\
& \leq \frac{2+\eta}{1+\eta}\lambda\sum_{k=1}^{K}w_k\|\Delta_k\|_\star.\label{eq:ap4}
\end{align}
From \eqref{eq:ap3}, we have
\begin{align}
\sum_{k=1}^{K}w_k\|\Delta_k\|_\star & \leq \sum_{k=1}^{K}w_k\|\Delta_k'\|_\star+\sum_{k=1}^{K}w_k\|\Delta_k''\|_\star \notag\\
& \leq \frac{2+2\eta}{\eta}\left(\sum_{k=1}^{K}w_k\sum_{j=r_k+1}^{m_k}\sigma(\bB_{0k},j)+\sum_{k=1}^{K}w_k\|\Delta_k'\|_\star\right) \notag\\
& \leq \frac{2+2\eta}{\eta}\left(\sum_{k=1}^{K}w_k\sum_{j=r_k+1}^{m_k}\sigma(\bB_{0k},j)+\sum_{k=1}^{K}\sqrt{2r_k}w_k\|\Delta_k'\|_\Fn\right). \label{eq:ap5}
\end{align}
The last inequality is due to the fact that $\|\Delta\|_\Fn=\|\Delta'\|_\Fn+\|\Delta''\|_\Fn$.

Now, combining \eqref{eq:ap2}, \eqref{eq:ap4} and \eqref{eq:ap5}, we know that on the event $\cap_{k=1}^{K}\mathcal{A}_k$, either $\|\Delta\|_{\Fn}\leq \delta$, or
\begin{align*}
\kappa(\X)\|\Delta\|_{\Fn}^2 & \leq
\frac{2(2+\eta)}{\eta}\lambda\left(\sum_{k=1}^{K}w_k\sum_{j=r_k+1}^{m_k}\sigma(\bB_{0k},j)+\sum_{k=1}^{K}\sqrt{2r_k}w_k\|\Delta_k\|_\Fn\right)\\
& \leq
\frac{2(2+\eta)}{\eta}\lambda\left(\sum_{k=1}^{K}w_k\sum_{j=r_k+1}^{m_k}\sigma(\bB_{0k},j)
+\|\Delta\|_\Fn\sqrt{2\sum_{k=1}^{K}r_kw_k^2} \right).
\end{align*}
That is,
\begin{align*}
\|\Delta\|_\Fn^2 \preceq \max\left\{\delta^2, \frac{\lambda^2\sum_{k=1}^{K}r_k w_k^2}{\kappa^2(\bX)}, \frac{\lambda\sum_{k=1}^{K}w_k\sum_{j=r_k+1}^{m_k}\sigma(\bB_{0k},j)}{\kappa(\bX)}\right\}.
\end{align*}
Lastly, from the proof of Theorem \ref{th:1}, choosing $\lambda=2(1+\theta) \tau$ ensures that
\begin{align*}
\mathbb{P}\{\cap_{k=1}^{K}\mathcal{A}_k\}
\geq 1 - \sum_{k=1}^{K}\exp\{-\frac{1}{2}\theta^2(q+r(\bX_k))\}.
\end{align*}
This completes the proof.

\end{proof}

\section{Additional Simulation with Correlated Errors}\label{sec:appendix:sim}
We conduct additional simulation studies where the errors in $\bE$ are correlated. In particular, we consider an AR(1) covariance structure with common variance 1 and autocorrelation 0.5 for the random errors in $\bE$ in Settings 1--5. The same methods are used and the results are shown in Table \ref{tab:AR1} and Figure \ref{fig:AR1}. The results are very similar to those with i.i.d.\ errors.
A closer look reveals that the proposed iRRR method is very robust against the violation of the independent error assumption, while other methods (especially MTL and grLasso) are more sensitive. 

\begin{table}[bh]
  \centering
  \caption{Simulation results for {\bf Settings 1--4} with correlated errors. The mean and standard deviation (in parenthesis) of MSPE over 100 simulation runs are presented. In each setting, the best results are highlighted in boldface.}\label{tab:AR1}
  \begin{tabular}{lrrrr}
    \hline
     & \multicolumn{1}{c}{iRRR} & \multicolumn{1}{c}{aNNP} & \multicolumn{1}{c}{RRRR} & \multicolumn{1}{c}{OLS} \\\hline
    {\bf Setting 1} & {\bf 7.74} (0.22) & 9.11 (0.32) & 10.27 (0.43) & 25.14 (0.58) \\\hline
    {\bf Setting 2} & {\bf 4.62} (0.10) & 5.63 (0.18) & 5.35 (0.14) & 25.14 (0.58) \\\hline
    \hfill ($r_0=20$) & 10.73 (0.26) & {\bf 9.10} (0.33) & 10.06 (0.45) & 25.17 (0.60) \\
    {\bf Setting 3} ($r_0=40$) & {\bf 13.10} (0.24) & 14.09 (0.33) & 15.08 (0.17) & 25.11 (0.60) \\
    \hfill ($r_0=60$) & {\bf 14.40} (0.23) & 16.43 (0.38) & 15.70 (0.16) & 25.16 (0.52) \\\hline
    \hfill($K=3$) & {\bf 11.03} (0.26) & 17.11 (0.48) & 17.63 (0.24) & 43.87 (0.84) \\
    {\bf Setting 4} \ ($K=4$) & {\bf 14.12} (0.25) & 25.33 (0.64) & 20.97 (0.22) & 68.06 (1.29) \\
    \hfill($K=5$) & {\bf 16.09} (0.29) & 30.78 (0.40) & 23.01 (0.22) & 101.81 (1.45) \\
    \hline
  \end{tabular}
\end{table}

\begin{figure}[bh]
  \centering
  \includegraphics[width=4in]{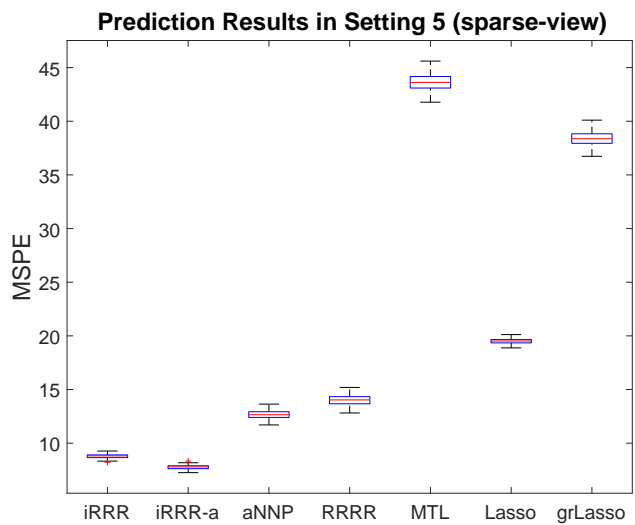}
  \caption{Simulation results for {\bf Setting 5 (sparse-view)} with correlated errors.}\label{fig:AR1}
\end{figure}

\clearpage
\bibliographystyle{chicago}


\end{document}

%% file: myCommand.tex
\definecolor{Pink}{rgb}{1.0, 0.5, 0.5}
\definecolor{Maroon}{rgb}{0.8, 0.0, 0.0}

\def\boxit#1{\vbox{\hrule\hbox{\vrule\kern6pt\vbox{\kern6pt#1\kern6pt}\kern6pt\vrule}\hrule}}

\newcommand{\bb}{\mbox{\bf b}}

\newcommand{\be}{\mbox{\bf e}}

\newcommand{\by}{\mbox{\bf y}}

\newcommand{\bA}{\mbox{\bf A}}
\newcommand{\bB}{\mbox{\bf B}}
\newcommand{\bC}{\mbox{\bf C}}

\newcommand{\bD}{\mbox{\bf D}}
\newcommand{\bE}{\mbox{\bf E}}

\newcommand{\bI}{\mbox{\bf I}}

\newcommand{\bL}{\mbox{\bf L}}
\newcommand{\bQ}{\mbox{\bf Q}}
\newcommand{\bR}{\mbox{\bf R}}

\newcommand{\bU}{\mbox{\bf U}}
\newcommand{\bV}{\mbox{\bf V}}

\newcommand{\bX}{\mbox{\bf X}}

\newcommand{\bY}{\mbox{\bf Y}}
\newcommand{\bZ}{\mbox{\bf Z}}

\newcommand{\bmu}{\mbox{\boldmath $\mu$}}

\newcommand{\tr}{\mathrm{tr}}

\def\beqn{\begin{eqnarray}}
\def\eeqn{\end{eqnarray}}
\def\beqns{\begin{eqnarray*}}
\def\eeqns{\end{eqnarray*}}

\def\0{{\bf 0}}

\def\I{{\bf I}}

\def\bQ{{\bf Q}}
\def\U{{\bf U}}
\def\V{{\bf V}}

\def\X{{\bf X}}

\def\Y{{\bf Y}}

\def\1{{\bf 1}}